\documentclass[a4paper,USenglish,thm-restate]{lipics-v2021}
\hideLIPIcs

\usepackage{graphicx}
\usepackage{cite}

\newtheorem{theoremalpha}{Theorem}

\newcommand{\I}{\mathcal{I}}
\renewcommand{\O}{\mathcal{O}}
\newcommand{\K}{\mathcal{K}}
\newcommand{\Cl}{\mathsf{Cl}}
\newcommand{\SA}{\mbox{\rm SA}}
\newcommand{\HX}{\mbox{\rm HX}}

\newcommand{\Bary}{\mathsf{Bary}}

\newcommand{\Skel}{\mathsf{Skel}}
\newcommand{\val}{\mathsf{val}}
\newcommand{\IIS}{{\rm\textsf{IIS}}}
\newcommand{\ID}{{\rm\textsf{ID}}}

\EventEditors{Rotem Oshman}
\EventNoEds{1}
\EventLongTitle{37th International Symposium on Distributed Computing (DISC 2023)}
\EventShortTitle{DISC 2023}
\EventAcronym{DISC}
\EventYear{2023}
\EventDate{October 10-12, 2023}
\EventLocation{L'Aquila, Italy}
\EventLogo{}
\SeriesVolume{281}
\ArticleNo{3}

\title{One Step Forward, One Step Back:\\
	FLP-Style Proofs and the Round-Reduction Technique for Colorless Tasks}

\titlerunning{FLP-Style Proofs and the Round-Reduction Technique for Colorless Tasks}
\author{Hagit Attiya}{Department of Computer Science, Technion, Israel}{hagit@cs.technion.ac.il}{0000-0002-8017-6457}{partially supported by the Israel Science Foundation (grants 380/18 and 22/1425)}
\author{Pierre Fraigniaud}{IRIF --- CNRS \& Universit\'e Paris Cit\'e}{pierre.fraigniaud@irif.fr}{}{partially supported by the ANR projects DUCAT (ANR-20-CE48-0006), FREDDA (ANR-17-CE40-0013), and QuData (ANR-18-CE47-0010)}
\author{Ami Paz}{LISN --- CNRS \& Universit\'e Paris-Saclay}{ami.paz@lisn.fr}{0000-0002-6629-8335}{}
\author{Sergio Rajsbaum}{IRIF, \'Ecole Polytechnique and Instituto de Matem\'aticas, UNAM, Mexico}{rajsbaum@im.unam.mx}{}{}
\authorrunning{Attiya, Fraigniaud, Paz and Rajsbaum}

\Copyright{Hagit Attiya, Pierre Fraigniaud, Ami Paz, and Sergio Rajsbaum} 

\ccsdesc[500]{Theory of computation~Distributed algorithms}

\keywords{Wait-free computing, lower bounds}

\acknowledgements{The authors thank Faith Ellen and the referees of DISC'23 for helpful comments}

\nolinenumbers 

\begin{document}

\maketitle


\begin{abstract}
The paper compares two generic techniques for deriving lower bounds and impossibility results in distributed computing.
First, we prove a \emph{speedup theorem} (a-la Brandt, 2019), for wait-free \emph{colorless} algorithms, aiming at capturing the essence of the seminal \emph{round-reduction} proof establishing a lower bound on the number of rounds for 3-coloring a cycle (Linial, 1992), and going by \emph{backward induction}.
Second, we consider
\emph{FLP-style} proofs, aiming at capturing the essence of the seminal consensus impossibility proof (Fischer, Lynch, and Paterson, 1985) and using \emph{forward induction}.

We show that despite their very different natures, these two forms of proof are tightly connected. In particular, we show that for every colorless task~$\Pi$, if there is a round-reduction proof establishing the impossibility of solving $\Pi$ using wait-free colorless algorithms, then there is an FLP-style proof establishing the same impossibility. For 1-dimensional colorless tasks (for an arbitrarily  number $n\geq 2$ of processes), we prove that the two proof techniques have exactly the \emph{same} power, and more importantly, both are \emph{complete}: if a 1-dimensional colorless task is \emph{not} wait-free solvable by $n\geq 2$ processes, then the impossibility can be proved by both proof techniques. Moreover, a round-reduction proof can be \emph{automatically} derived, and an FLP-style proof can be automatically generated from it.

Finally, we illustrate the use of these two techniques by establishing the impossibility of solving any colorless \emph{covering} task of arbitrary dimension by wait-free algorithms.
\end{abstract}

\newpage
\tableofcontents
\newpage

\section{Introduction}

We analyze the relative power of two generic and versatile techniques for establishing lower bounds and impossibility results in asynchronous distributed computing.
We focus on solving \emph{tasks}, defined as triples $\Pi=(\I,\O,\Delta)$,
where processes start with initial input values defined by $\I$,
and decide irrevocably on output values allowed by $\O$ after communicating
with each other for some number of steps, respecting the input/output relation $\Delta$;
the sets of processes, input values and output values are all finite.

This paper concentrates on the family of
\emph{colorless} tasks, including consensus, set agreement~\cite{Chaudhuri93},
loop agreement~\cite{HerlihyR03},
and various robot and graph agreement tasks~\cite{AlistarhER21,CastanedaRR16}.
A colorless task is defined only in terms
of input and output values,
regardless of the number of processes involved, and regardless of which process
has a particular input or output value;
accordingly, $\I$ and $\O$ consist of sets of values, without process ids.
For instance, in the binary consensus task, $\I=\big\{\{0\},\{1\},\{0,1\}\big\}$, meaning that all processes may start with input~0, or all processes may start with input~1, or some processes may start with input~0 while others may start with input~1. In the same task, $\O=\big\{\{0\},\{1\}\big\}$, meaning that the only valid output configurations are when all processes output~0, or all processes output~1.
Finally, consensus specification is captured by $\Delta(\{0\})=\{0\}$, $\Delta(\{1\})=\{1\}$, and $\Delta(\{0,1\})=\big\{\{0\},\{1\}\big\}$, meaning that, if there was an initial agreement between the input values then the processes must stick to this agreement and output their input values, and otherwise they are allow to output either~0 or~1, as long as they all agree on the same value.

Colorless tasks are simpler to analyze than general tasks,
such as symmetry breaking tasks~\cite{CastanedaIRR16},
but they are still undecidable~\cite{HerlihyR97}.
They have an elegant computability
characterization~\cite{bookHerlihyKR2013,HerlihyRRS17},
essentially stating that a colorless task is wait-free solvable
if and only if there is a continuous map from the geometric realization
of $\I$ to that of $\O$ that respects $\Delta$.
One major interest of colorless tasks is that, whenever solvable, they can be solved by simple algorithms, referred to as \emph{colorless algorithms}~\cite{HerlihyRRS17} (see also~\cite[Ch.~4]{bookHerlihyKR2013}).
Roughly speaking, such algorithms ignore multiplicities of input values and process states, and manipulate only \emph{sets} of values; in contrast, general algorithms manipulate \emph{vectors} of values and take into account which processes possess which value.

The two lower-bound or impossibility techniques at the core of our work are \emph{FLP-style proofs}, named after Fischer, Lynch and Paterson~\cite{FLP85}, and \emph{round-reduction proofs}, whose first occurrence might be attributed to Linial~\cite{linial92}.
The former offers a form of \emph{forward} induction technique, while the latter offers a form of \emph{backward} induction, and they both present some form of \emph{locality}.
We recall these two techniques hereafter.

\subsection{Stepping  Forward: FLP-Style Impossibility Proofs}

The celebrated FLP proof technique~\cite{FLP85} is perhaps the most used technique
for proving impossibility results in distributed computing.
It has been used to prove the impossibility of solving consensus
and several other problems~\cite{AttiyaEllenBook,FichR03hundred,Lynch89hundred},
as well as to derive lower bounds~\cite{AttiyaEllenBook,AguileraT99,KeidarR03}.
To prove that a task $\Pi=(\I,\O,\Delta)$ is not wait-free solvable,
the FLP technique considers a hypothetical algorithm \textsc{alg} solving the task,
and constructs an infinite sequence $\sigma_0,\sigma_1,\ldots$ of system configurations such that, for every $i\geq 0$,
\begin{itemize}
\item $\sigma_{i+1}$ is a successor of $\sigma_i$, and
\item \textsc{alg} cannot output in $\sigma_i$.
\end{itemize}
To initiate this sequence, the prover is allowed to ask the \emph{valencies}
of all initial configurations~$\sigma\in \I$,
where the valency of a configuration $\sigma$ is the set of values that are output by \textsc{alg}
in all executions starting from $\sigma$.
Based on these valencies, the prover selects an initial configuration $\sigma_0\in \I$.
Then, given a configuration $\sigma_i$, the prover asks the algorithm for
the valencies of some successor configurations of $\sigma_i$, and one of these configurations is selected to be the next configuration $\sigma_{i+1}$ in the sequence, and so on.
This is a form of \emph{forward induction},
starting with $\sigma_0$ and constructing a sequence of configurations one after the other.

If \textsc{alg} actually solves the task~$\Pi$, then the prover will fail to construct an infinite sequence, merely because \textsc{alg} can reveal the actual valencies that it produces for each given configuration.
On the other hand, if for every algorithm \textsc{alg} hypothetically solving~$\Pi$ the prover successfully constructs an infinite sequence, then this establishes that $\Pi$ is not solvable because the sequence is constructed in such a way that  \textsc{alg}  cannot output in $\sigma_i$.
For instance, the FLP impossibility proof for consensus~\cite{FLP85} constructs such an infinite sequence $\sigma_0,\sigma_1,\ldots$ for every algorithm \textsc{alg} by establishing that (1)~there exists an initial \emph{bivalent} configuration~$\sigma_0$ (i.e., a configuration from which an execution of  \textsc{alg} leads all processes to output~0, and another execution of  \textsc{alg} leads all processes to output~1), and (2)~for every bivalent configuration~$\sigma_i$, there exists a bivalent one-step successor~$\sigma_{i+1}$ of $\sigma_i$ that is bivalent.

An FLP-style impossibility proof is a simple case of
extension-based impossibility proofs~\cite{AlistarhAEGZ19}
	and \emph{local proofs}~\cite{AttiyaCR20},
	hence if such a proof exists for a task~$\Pi$ then there are also extension-based and local impossibility proofs for it.
	All these techniques use types of \emph{valency-arguments}~\cite[Chapter~7]{AttiyaEllenBook}, in the sense that they consider output values in executions starting from a given configuration, and then decide on the next configuration.
All these techniques work for consensus but fail for set agreement (i.e., set agreement is not solvable, but this impossibility cannot be established by valency arguments),
and the exact set of tasks for which each of them applies is not known.

\subsection{Stepping Back: Round-Reduction Impossibility Proofs}
\label{subsec:rr}

The round-reduction proof technique in distributed computing can be traced back to the seminal work of  Linial~\cite{linial92} establishing a lower bound for coloring the $n$-node cycle~$C_n$ using a (failure-free) synchronous algorithm. In a nutshell, he proved that for every $t\geq 1$, if there exists a $t$-round algorithm \textsc{alg} producing a proper $k$-coloring of $C_n$, then there exists a $(t-1)$-round algorithm $\textsc{alg}'$ producing a proper $2^{2^k}$-coloring of $C_n$.
Repeating this argument for roughly $\frac12\log^\ast n$ times implies that if there is a
$(\frac12\log^\ast n)$-round algorithm for $3$-coloring $C_n$ then there exists a $0$-round algorithm for $(n-1)$-coloring it, which is impossible.
Here, $\log^\star n$ denotes the number of times one should apply the $\log_2$ function to get from $n$ to a value smaller than~1.

This technique was generalized as a \emph{speedup theorem} by Brandt~\cite{Brandt19}.
Such a theorem is based on establishing the existence of a map~$F$ transforming any task~$\Pi$ in some class~$\mathcal{T}$ of tasks into another task $\Pi'\in\mathcal{T}$ such that, for every $t\geq 1$, if $\Pi$ is solvable in $t$ rounds by an algorithm \textsc{alg} from some class~$\mathcal{A}$ of algorithms, then $F(\Pi)$ is solvable in $t-1$ rounds by an algorithm  $\textsc{alg}'\in\mathcal{A}$.
Whenever such a theorem can be established, we get that for every task $\Pi\in\mathcal{T}$, and for every $t\geq 0$, if the task $F^{(t)}(\Pi)$ obtained by iterating $t$~times $F$ on~$\Pi$ is not solvable in zero rounds by an algorithm in~$\mathcal{A}$, then $\Pi$ is not solvable in $t$ rounds by an algorithm in~$\mathcal{A}$, which provides a lower bound on the complexity of~$\Pi$.
In particular, if $F^{(t)}(\Pi)$ is not solvable in zero rounds  by an algorithm in~$\mathcal{A}$ for all $t\geq 0$, then $\Pi$ cannot be solved by an algorithm in~$\mathcal{A}$.
We refer to this technique as a \emph{round-reduction} impossibility or and lower bound proof.
In contrast with the forward induction approach of FLP-style proofs, round-reduction is an a-posteriori technique, starting by assuming an algorithm solves a problem in $t$ rounds, claiming another problem is solvable in $t-1$ rounds, and repeating this argument down to $0$ rounds;
we hence refer to it as backward induction.

A speedup theorem has been established in~\cite{Brandt19} for solving locally-checkable labeling (LCL) tasks~\cite{NaorS95} using algorithms running in the anonymous \textsf{LOCAL} model~\cite{Peleg00}.
A speedup theorem has also been established in~\cite{FPR2022}, but for general (colored) tasks and wait-free algorithms running in the \emph{iterated immediate snapshot} (\IIS\/) model~\cite{AWbook}.
The transformations $F_{\mbox{\tiny \textsf{LOCAL}}}$ and  $F_{\mbox{\tiny \IIS}}$ used in \cite{Brandt19} and \cite{FPR2022} respectively, are of very different natures.
Nevertheless, both enabled to establish lower bounds for various tasks, including sink-less orientation and maximal independent set (MIS) in the  \textsf{LOCAL} model (see \cite{Balliu0HORS19,Brandt19}), and approximate agreement  in the \IIS\/ model, and even when the \IIS\/ model is enhanced with powerful objects like \textsf{test\&set} (see~\cite{FPR2022}).

Our first contribution is a speedup theorem for wait-free colorless algorithms solving colorless tasks in the \IIS\/ model. It is important to note that, although the set of colorless tasks is a subset of the set of general tasks, the speedup theorem in~\cite{FPR2022}  does not imply our speedup theorem, since the transformation $F_{\mbox{\tiny IIS}}$ used in~\cite{FPR2022} does not apply to colorless \emph{algorithms}.
Specifically, if $\Pi$ is solvable in $t$ rounds by $n$ processes running a wait-free colorless algorithm, then $F_{\mbox{\tiny IIS}}(\Pi)$ is indeed solvable in $t-1$ rounds by $n$ processes, but running a wait-free algorithm that \emph{may not be colorless}. As a consequence, the theorem cannot be iterated, which ruins the ability to design a round-reduction proof for colorless tasks.
The speedup theorem for colorless tasks we present here uses a transformation that preserves solvability by colorless algorithms.
The transformation $\Cl$ we define (Definition~\ref{def:colorless-closure}) has the following properties, as shown in Theorem~\ref{theo:closure} and Theorem~\ref{theo:closure-reciprocal}.
Applications of this theorem to approximate agreement and covering tasks can be found in Section~\ref{sec:applications}.

\begin{theoremalpha}\label{theo:A}
For every $n\geq 2$ and every $t\geq 1$,
the transformation $\Cl$ maps any colorless task $\Pi$ to a colorless task $\Cl(\Pi)$ such that, when considering $n$ processes running a colorless wait-free algorithm in the \IIS\/ model:
\begin{itemize}
	\item
	If $\Pi$ is $1$-dimensional, then
	$\Pi$ is solvable in $t$ rounds
	if and only if
	$\Cl(\Pi)$ is solvable in $t-1$ rounds;
	
	\item
	Regardless of the dimension of $\Pi$,
	if $\Pi$ is solvable in $t$ rounds
	then
	$\Cl(\Pi)$ is solvable in $t-1$ rounds.
%
%
\end{itemize}
\end{theoremalpha}

A round-reduction impossibility or lower bound proof derived from Theorem~\ref{theo:A} essentially proceeds by computing $\Cl^{(t)}(\Pi)$, and checking whether $\Cl^{(t)}(\Pi)$ is solvable in zero rounds.
If the answer is negative for some $t\geq 1$, the proof successfully shows that $\Pi$ cannot be solved in $t$ rounds, and if it is negative for all $t\geq 1$, the proof shows $\Pi$ is not solvable;
otherwise, the proof fails.
While $\Cl$ is not the only possible round-reduction operator for colorless tasks, we focus solely on it in this work.

In Appendix~\ref{app:comparison-based}, we illustrate the fact that round-reduction does not extend in a straightforward manner to all classes of algorithms.
To this end, we consider \emph{comparison-based} algorithms, an important class of algorithms used for studying tasks such as renaming and weak symmetry-breaking.
We show that a closure operator similar to the ones considered here and in~\cite{FPR2022} but restricted to comparison-based algorithms, does not suffice for deriving a speedup theorem for comparison-based algorithms.

\subsection{Round-Reduction vs.~FLP-Style Impossibility Proofs}

Interestingly, as for extension-based proofs, the round-reduction proofs derived from the speedup theorem in~\cite{FPR2022} work for consensus 
but fail for set-agreement 
and the same holds for our transformation~$\Cl$.
We next show why the fact that both transformations succeed for binary consensus but fail for set-agreement should not come as a surprise, in light of prior results about FLP-style proofs.

Our second contribution shows that, although round-reduction proofs and FLP-style proofs may appear very different in nature, their power in term of establishing impossibility results can be compared.
We actually show that the FLP-style proof technique is at least as strong as the round-reduction proof technique.

\begin{theoremalpha}%
	[Theorem~\ref{theo:rr-to-ex-colorless}]
	\label{theo:B}
	For every colorless task~$\Pi$ and $n\geq2$, if there is a round-reduction  proof establishing the impossibility of solving $\Pi$ by $n$ processes running a wait-free colorless algorithm, then there is an FLP-style proof establishing the same impossibility.
\end{theoremalpha}

So, in particular, the fact that there is no round-reduction proof for the impossibility of solving set-agreement wait-free should not come as a surprise, since it is known that set-agreement has no extension-based impossibility proof~\cite{AlistarhAEGZ19,AttiyaCR20,BrusseE21}, which is a proof technique at least as strong as FLP-style proofs.
Yet, the situation is a bit subtle here: the results in~\cite{AlistarhAEGZ19,AttiyaCR20} rule out the existence of an FLP-style impossibility proof for solving set-agreement using general algorithms,
and~\cite{BrusseE21} proves a similar result for anonymous algorithms (where processes see a multi-set of the values stored in the memory and not a vector),
but these works do not necessarily rule out the existence of an FLP-style impossibility proof for solving set agreement using colorless algorithms.
The restriction to colorless algorithm allows the claimed algorithm \textsc{alg}
less flexibility regarding the valencies it announces to the prover, so an FLP-style impossibility proof when restricting \textsc{alg} to be colorless is not ruled out by the previous results.

Nevertheless, we were able to show that, for \emph{1-dimensional} colorless tasks, round-reduction proofs and FLP-style proofs have exactly the same power
(Corollary~\ref{cor:flp iff rr dim 1}).
Recall that, in 1-dimensional colorless tasks, $\I$ and $\O$ are graphs,
i.e., the $n$ processes can start with at most two different input values and output at most two different values.
This family includes binary consensus,
approximate agreement with binary inputs, and the colorless version of wait-free checkable tasks~\cite{FraigniaudRT13}, but not set-agreement.
1-dimensional colorless tasks are well studied, and they are known to be wait-free solvable if and only if they are $1$-resilient solvable,
both in the read/write model (with at least two processes)
and in the message-passing model (with at least three processes)~\cite{BiranMZ90,BorowskyGLR01,HerlihyR12}.
In fact, we not only show equivalence between round-reduction and FLP-style proof techniques for 1-dimensional colorless tasks, but we show that both are complete for these tasks
(Corollary~\ref{cor:round-reduction-complete} for the round-reduction technique
and
Corollary~\ref{cor:flp is complete for 1 dim} for FLP-style proofs).
Hence, if a 1-dimensional colorless task is not wait-free solvable by $n\geq 2$ processes, then this impossibility can be proved by both proof techniques.

\begin{theoremalpha}\label{theo:C}
The round-reduction and FLP-style proof techniques are both complete for 1-dimensional colorless tasks and wait-free colorless algorithms.
\end{theoremalpha}

Theorem~\ref{theo:C} follows from the fact that, for 1-dimensional colorless tasks (and colorless algorithms), our speedup theorem (Theorem~\ref{theo:A}) also provides a \emph{necessary} condition, that is, for every 1-dimensional colorless task $\Pi$, every $n\geq 2$, and every $t\geq 1$, $\Pi$ is solvable in $t$ rounds by $n$ processes running a wait-free colorless algorithm in the \IIS\/ model \emph{if and only if} $\Cl(\Pi)$ is solvable in $t-1$ rounds by $n$ processes running a wait-free colorless algorithm in the \IIS\/ model.
This if-and-only-if condition provides a mechanical way for deciding whether a given 1-dimensional colorless task $\Pi$ is solvable.
Indeed, the transformation $\Cl$ used in Theorem~\ref{theo:A} has a desirable property: it preserves the number of input and output values (i.e., the vertices in $\I$ and $\O$), and it may just potentially add some combinations of output values that were not legal in~$\Pi$ but become legal in $\Cl(\Pi)$.
As a consequence, iterating $\Cl$ starting from $\Pi$ necessarily leads to a fixed point $\Pi^\ast$ for $\Cl$, i.e., $\Cl(\Pi^\ast)=\Pi^\ast$,
after a bounded number of iterations.
It follows that a 1-dimensional colorless task $\Pi$ is wait-free solvable in the \IIS\/ model if and only if $\Pi^\ast$ is wait-free solvable in zero rounds, which is decidable.  The completeness of the FLP-style proof technique follows. Indeed, if a 1-dimensional colored task $\Pi$ is not wait-free solvable, then $\Pi$ has a round-reduction impossibility proof (by computing the fixed point $\Pi^\ast$, and showing that $\Pi^\ast$  is not solvable in zero rounds), from which it follows, thanks to Theorem~\ref{theo:B}, that $\Pi$ has an FLP-style impossibility proof.

\subsection{Applications}

We illustrate the concepts and results introduced in this paper by applying them to the vast class of \emph{covering tasks},
whose colored version was introduced and studied in~\cite{FraigniaudRT13}
under the name \emph{locality-preserving} tasks
	and further studied in~\cite{DitmarschGLLR21}.
To get the intuition of such a tasks, it is easier to consider 1-dimensional covering tasks, for which $\I$ and $\O$ are graphs.
Recall that for two connected simple (i.e., no self-loops nor multiple edges) graphs $G$ and $H$, and a function $f:V(H)\to V(G)$, the pair $(H,f)$ is a covering of~$G$ if $f$ is an homorphism (i.e., it preserves edges) and, for every $v\in V(H)$,
the restriction of $f$ to $N_H[v]$ is a one-to-one mapping $f:N_H[v]\to N_G[f(v)]$.
For instance, for $C_3=(v_0,v_1,v_2)$, $C_6=(u_0,\dots,u_5)$, and $f:V(C_6)\to V(C_3)$ defined as $f(u_i)=v_{i\bmod 3}$, $(C_6,f)$ is a covering of $C_3$. A covering $(\O,f)$ of $\I$ induces a task $\Pi=(\I,\O,\Delta)$ where $\Delta$ is essentially defined as~$f^{-1}$. For higher dimensional colorless tasks,  $\I$ and $\O$ are connected simplicial complexes, and $f$ must be simplicial, but the general idea is the same~\cite{Rotman73}. A covering $(\O,f)$ of $\I$  is trivial if  $\I$ and $\O$  are isomorphic.
It is known that no non-trivial covering tasks can be solved wait-free in the \IIS\/ model~\cite{FraigniaudRT13}.
Here, we consider a colorless version of covering tasks, and study impossibility proofs for solving them using colorless algorithms.

\begin{theoremalpha}\label{theo:D}
Every non-trivial covering task admits an FLP-style impossibility proof for $n\geq2$ processes running wait-free colorless algorithms.
\end{theoremalpha}

The proof is based on showing that, for every covering task~$\Pi$, the transformation $F$ used in our speedup theorem satisfies $F(\Pi)=\Pi$, i.e., $\Pi$ is itself a fixed point for~$F$.
As a consequence, since $\Pi$ is not solvable in zero rounds (unless $\I$ and $\O$ are isomorphic, i.e., the task~$\Pi$ is trivial), there is a round-reduction impossibility proof for~$\Pi$
(Theorem~\ref{thm:covering impossibility}),
and the existence of an FLP-style impossibility proof for~$\Pi$ then follows from Theorem~\ref{theo:B}.
This last fact is of particular interest, as it shows that FLP-style proofs are not limited to cases where $\O$ is disconnected
or to problems that are unsolvable even when restricted to a single input simplex and its faces (as is the case for consensus and approximate agreement).

%


\section{Model and Definitions}
\label{sec:model}

We consider the  wait-free  iterated immediate snapshots model (IIS).
This model and its variants have been frequently used e.g.~\cite{BorowskyG97,iterated2010,KRH18} due to its simplicity, while
being equivalent to the usual wait-free read/write shared memory model for task solvability~\cite{GafniR10,BorowskyG97}.
Furthermore, it is known that as far as colorless task solvability is concerned, one can assume
\emph{colorless
computation} without loss of generality~\cite{bookHerlihyKR2013,HerlihyRRS17}, by which we mean
that in each round of computation processes do no consider which process wrote a value, nor by how many processes it was written.
We next provide a brief overview of the model.

Processes communicate through a sequence of shared memory objects, and the computation is split into rounds.
In the $i$-th round, process $p$ writes a value to the $p$-th location of the $i$-th object, and then takes a snapshot of all the $i$-th object.
The \emph{view} of a process at the end of a round is the set of values it read from the object, without the information of which process wrote which value, nor by how many processes it was written to the memory.
As common in lower bound proofs, we only consider full information protocols: in each round, each process writes its entire state (or equivalently, the view read from the previous object) to the memory.

\subsection{Colorless Tasks}

In this paper we consider colorless tasks~\cite{LynchR96}.
See~\cite[Chapter 4]{bookHerlihyKR2013} and~\cite{HerlihyRRS17} for an overview.

\begin{definition}
A \emph{colorless} task $\Pi=(\I,\O,\Delta)$ is defined by an input simplicial complex $\I$, an output simplicial complex $\O$, and an input-output specification $\Delta:\I\to 2^\O$ mapping every simplex $\sigma\in\I$ to a sub-complex $\Delta(\sigma)$ of~$\O$ with dimension at most~$\dim(\sigma)$.
\end{definition}

The semantics of a colorless task $\Pi=(\I,\O,\Delta)$ is that every vertex of $\I$ is an input value, and every vertex of $\O$ is an output value. A set of processes may start with different input values in~$V(\I)$, as long as the set~$\sigma$ of these input values belongs to~$\I$. To solve the task, it is required that any collection of processes starting with input values forming a set~$\sigma\in\I$  outputs only output values in the vertices $V(\O)$, as long as the set $\tau$ of these output values forms a simplex in $\Delta(\sigma)$.

\subsection{Examples} 
The next two colorless tasks will serve as running examples in the rest of the paper.

\paragraph*{The Set Agreement Tasks}
Set agreement is a well studied relaxation of the classical consensus task.
Let $k\geq 2$. \emph{Set-agreement} with input set $[k]=\{1,\dots,k\}$ is the colorless task $\SA_k=(\I,\O,\Delta)$
where $\I=\{\sigma\subseteq [k]\mid \sigma\neq\varnothing\}$, $\O=\{\tau\subseteq [k]\mid (\tau\neq\varnothing) \land (\tau\neq [k])\}$, and, for every $\sigma\in \I$,
\[
\Delta(\sigma)=\{\tau\in \O\mid \tau\subseteq \sigma\}.
\]
Said differently, $\Delta(\sigma)=\{\tau\subseteq\sigma\}$ if $\dim(\sigma)<k-1$, and $\Delta(\sigma)=\O$ otherwise. 
$\SA_k$ is solvable (in zero rounds) for $n<k$ processes, but not solvable for $n\geq k$ processes~\cite{HerlihyS99,SaksZ93,BorowskyG93}.

\paragraph*{The Hexagone Task}
\begin{figure}[tb]
	\centering
	\includegraphics[width=8cm]{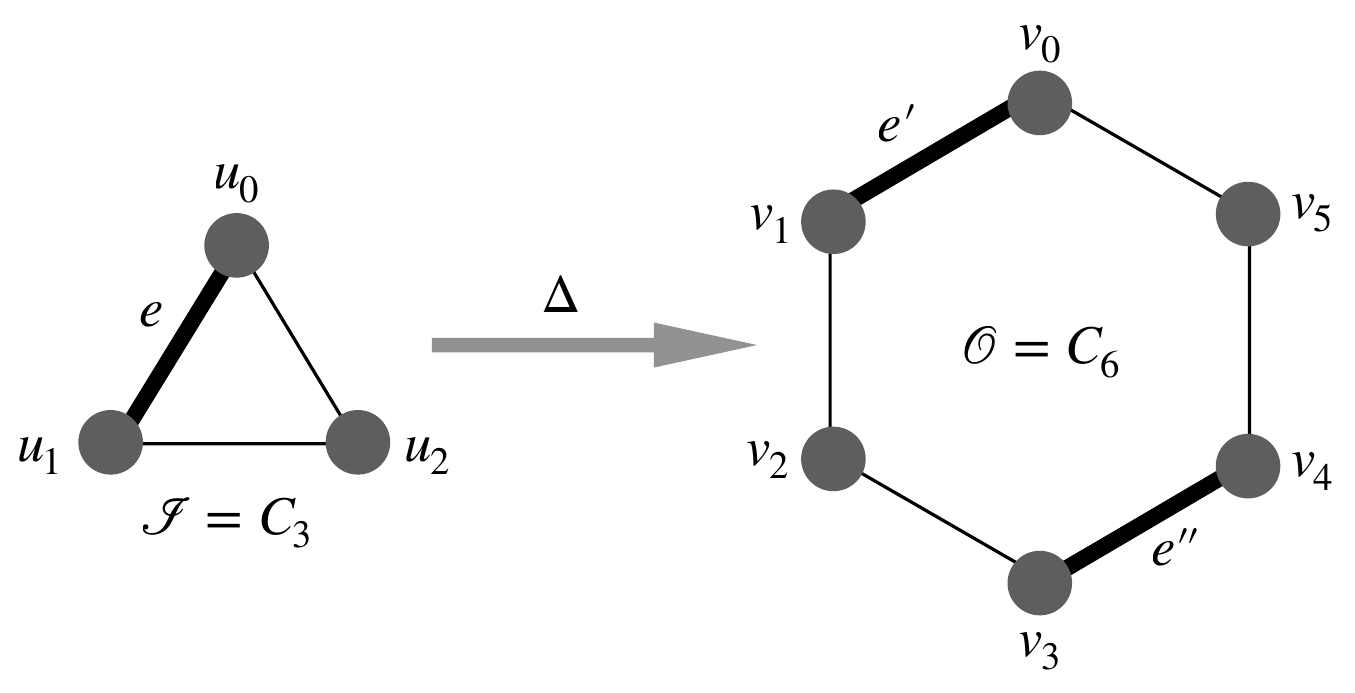}
	\caption{\sl The Hexagone Task. In particular, $\Delta(e)$ is the complex $e'\cup e''$.}
	\label{fig:hexagone}
\end{figure}
The Hexagone Task is one basic example of the colorless \emph{covering} tasks, thoroughly studied in Section~\ref{subsec:covering-tasks}.
For every $t\geq 3$, let $C_t$ denotes the $t$-node cycle. 
The \emph{hexagone task} is the task $\HX=(C_3,C_6,\Delta)$ where $C_3=(u_0,u_1,u_2)$, $C_6=(v_0,v_1,\dots,v_5)$
and $\Delta$ is defined as follows
(see Figure~\ref{fig:hexagone}). 
For every $i\in\{0,1,2\}$,
\[
\Delta(u_i)=\{\{v_i\},\{v_{i+3}\}\} \;\mbox{and}\; \Delta(\{u_i,u_{i+1 \bmod 3}\})=\{\{v_i,v_{i+1}\},\{v_{i+3},v_{i+4\bmod 6}\}\}, 
\]
where formally $\Delta(\{u_i,u_{i+1 \bmod 3}\})$ also contains all the vertices contained it its edges, i.e. $\{v_i\},\{v_{i+1}\},\{v_{i+3}\},\{v_{i+4\bmod 6}\}$.
The map $\Delta$ is the inverse of the maps
$f:V(C_6)\to V(C_3)$ defined as $f(u_i)=v_{i\bmod 3}$, 
where $(C_6,f)$ is a covering of $C_3$.

\subsection{Colorless Algorithms}

The solvability of a colorless task may depend on the number $n$ of processes involved in the computation.
Remarkably, it is enough to consider \emph{colorless computation} for colorless tasks, according to the following result (see, e.g.,~\cite{HerlihyRRS17}), that summarizes and formalizes what we need to know about the model of computation\footnote{
Results similar to this one are known for several models of computation~\cite{bookHerlihyKR2013}.
For instance,
a corresponding result is known for Byzantine failures~\cite[Theorem~6.5.1]{bookHerlihyKR2013} and dependent failures~\cite[Theorem 4.3]{HR2010adversaries}
with an adversary of core size $c=2$.
For dependent failures in the case of wait-free computation, the theorem states that a colorless task is solvable if and only if there exists a continuous map between geometric realization $f:\Skel_{c-1}(\I)\rightarrow \O$  carried by~$\Delta$.
}.
The result holds both for the wait-free read/write memory model and for its iterated version,
by the equivalence proved in~\cite{GafniR10}.

\begin{figure}[!htb]
\centering
\includegraphics[width=8cm]{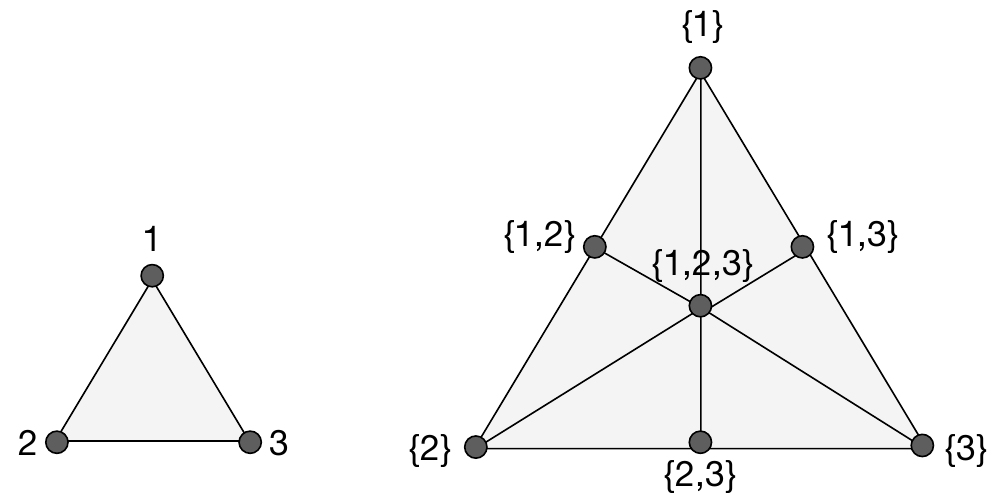}
\caption{\sl Barycentric subdivision}
\label{fig:barycentric}
\end{figure}

\begin{lemma}\label{lem:colorless-solvability}
A colorless task $\Pi=(\I,\O,\Delta)$ is read/write solvable by $n$~processes running a wait-free algorithm if and only if there exists $t\geq 0$ and a simplicial map
\[
f:\Bary^{(t)}(\Skel_n(\I))\to \O
\]
that agrees with~$\Delta$, i.e., for every $\sigma\in\I$, $f(\Bary^{(t)}(\Skel_n(\sigma)))\subseteq\Delta(\sigma)$.
Furthermore, $\Pi$ is $1$-round solvable by $n$ processes running a wait-free colorless algorithm in the IIS model if and only if there is a simplicial map
\[
f:\Bary(\Skel_n(\I))\to \O
\]
that agrees with~$\Delta$.
\end{lemma}

In the above, $\Bary^{(t)}$ denotes $t$~successive applications of the barycentric subdivision (see Fig.~\ref{fig:barycentric}), and $\Skel_n$ denotes the $(n-1)$-dimensional skeleton operator, i.e., $\Skel_n(\I)$ is the subcomplex of~$\I$ (resp., of~$\sigma$) including all simplices of~$\I$ (resp., all faces of~$\sigma$) with dimension at most~$n-1$.
When considering colored (general) tasks, 
a similar result holds when using chromatic simplicial maps and the \emph{standard chromatic subdivision};
for colorless tasks, however, the barycentric subdivision suffices since if a colorless task is solvable then it is solvable by an algorithm ignoring multiplicities of identical views (or inputs) in the snapshots~\cite{HerlihyRRS17}.


\section{Round-Reduction Proofs}
\label{sec:round-reduction-proofs}

This section is essentially dedicated to the proof of Theorem~\ref{theo:A}. 
The task transformation $\Cl$ in this theorem is based on a \emph{closure} operator similar to the one in~\cite{FPR2022}, where  
the main difference is that our closure operator is required to be colorless, and we stress that our result is not implied by the (colored) theorem of~\cite{FPR2022}.
That is, given a class $\mathcal{A}$ of algorithms, requiring the closure operator to be in the class may not be sufficient for extending the speedup theorem in~\cite{FPR2022} to apply to algorithms in~$\mathcal{A}$. We illustrate this in Appendix~\ref{app:comparison-based} for \emph{comparison-based} algorithms, an important class of algorithms used for studying tasks such as renaming and weak symmetry-breaking. 
Using a comparison-based closure operator does not suffice for deriving a speedup theorem for comparison-based algorithms. 
On the other hand, in this section we show that for colorless algorithms, restricting the closure operator to be colorless suffices.

\subsection{Colorless Closure}
\label{subsec:closure}

\subsubsection{Definition of Closure}

We first rephrase the notions of local tasks introduced in~\cite{FPR2022} in the context of colorless tasks.

\begin{definition}
Let $\Pi=(\I,\O,\Delta)$ be a colorless task, let $\sigma\in\I$, and let $\tau\subseteq V(\Delta(\sigma))$. Let us consider $\tau$ as a simplicial complex (with a unique facet). The \emph{local task}  with respect to $\sigma$ and $\tau$ is the colorless task $\Pi_{\tau,\sigma}=(\tau,\Delta(\sigma),\Delta_{\tau,\sigma})$ where, for every face $\tau'$ of $\tau$,
\[
\Delta_{\tau,\sigma}(\tau')=
\left\{\begin{array}{ll}
v & \mbox{if $\tau'=v$ is a vertex,}\\
\Skel_{\dim(\tau')}(\Delta(\sigma)) & \mbox{otherwise.}
\end{array}\right.
\]
\end{definition} 

Note that $\Pi_{\tau,\sigma}$ is a well-defined colorless task, because $\Skel_{\dim(\tau')}$ guarantees that the output complex for $\tau'$ has dimension at most $\dim(\tau')$;
this is true even if $\tau$ is not a simplex of $\Delta(\sigma)$, as seen in the second example in Section~\ref{examples: colorless closure}. 
Note also that the validity constraint of $\Pi_{\tau,\sigma}$ is just that if all processes start with the same input~$v$, i.e., if they start from the same vertex $v\in V(\tau)$, then they must all output~$v$. 
Without this constraint, the processes are only constrained to output values forming a legal set $\sigma'$ of outputs w.r.t.~$\sigma$, i.e., a set $\sigma'\in\Delta(\sigma)$ (and $\dim(\sigma')\leq \dim(\tau')$).

As opposed to the general (chromatic) tasks, which each has a fixed maximum number of participating processes, colorless tasks are defined for any number of processes. In particular, a colorless task may be solvable for a certain number of processes, but not for another number of processes. The definition below refers to solving local tasks with a prescribed number of processes.

\begin{definition}\label{def:colorless-closure}
The \emph{colorless closure} of a colorless task $\Pi=(\I,\O,\Delta)$ is the colorless task $\Cl(\Pi)=(\I,\O',\Delta')$ where $V(\O')=V(\O)$, and, for every~$\sigma\in \I$, and every non-empty set $\tau\subseteq V(\O)$, we set $\tau\in \Delta'(\sigma)$ if $\tau\subseteq V(\Delta(\sigma))$ and $\Pi_{\tau,\sigma}$ is solvable in one round by $\dim(\tau)+1$ processes running a colorless algorithm. The simplices of $\O'$ are the images of~$\Delta'$, and all their faces.
\end{definition}
Note that this definition is constructive, i.e., one can check the $1$-round solvability of $\Pi_{\tau,\sigma}$, by Lemma~\ref{lem:colorless-solvability}.
The closure operator also has the nice property that it does not change the allowed input and output values, and the only difference between $\Pi$ and $\Cl(\Pi)$ is in the addition of some allowed combinations of output values.

For a simplex $\tau\in\Delta(\sigma)$, the local task $\Pi_{\tau,\sigma}$ is solvable in 0~rounds, by having each process starting with input $v\in V(\tau)$ output~$v$. 
It follows that if $\tau\in\Delta(\sigma)$ then $\tau\in\Delta'(\sigma)$, and therefore, for every $\sigma\in \I$, $\Delta(\sigma)\subseteq \Delta'(\sigma)$.
As a consequence, the colorless closure $\Cl(\Pi)$ of a task~$\Pi$ is not more difficult to solve than $\Pi$. Whether or not $\Cl(\Pi)$ is \emph{simpler} to solve than $\Pi$ is one of the foci of this paper. Before going further, we establish hereafter that the input-output specification $\Delta'$ of the colorless closure of a colorless task is a \emph{carrier} map, which is a condition that is often required for a task~\cite{bookHerlihyKR2013}.

\subsubsection{The I/O-Specification of a Colorless Closure is a Carrier Map}

Let $\Pi=(\I,\O,\Delta)$ be a colorless task. Recall that $\Delta$ is a \emph{carrier} map if, for every $\sigma$ and $\sigma'$ in~$\I$,
$
\sigma\subseteq\sigma'\Longrightarrow \Delta(\sigma)\subseteq\Delta(\sigma'),
$
where the latter inclusion must be read as $\Delta(\sigma)$ is a subcomplex of $\Delta(\sigma')$.
Being a carrier map is not necessary for $\Pi$ to be solvable. However, if two simplices $\sigma$ and $\sigma'$ in~$\I$ satisfy $\sigma\subseteq\sigma'$ and $\Delta(\sigma)\smallsetminus \Delta(\sigma')\neq\varnothing$, the output values outside $\Delta(\sigma')$ cannot be used for a set of processes starting with input~$\sigma'$ since these output values cannot be extended in case processes with inputs in~$\sigma'\smallsetminus\sigma$ eventually participate later.
Therefore, we may as well remove all simplices from $\Delta(\sigma)$ that are not in $\Delta(\sigma')$, and restrict ourselves to input-output specifications that are carrier maps.

\begin{lemma}
	\label{lem: closure is carrier}
Let $\Pi=(\I,\O,\Delta)$ be a colorless task, and let $\Cl(\Pi)=(\I,\O',\Delta')$. If $\Delta$ is a carrier map then $\Delta'$ is a carrier map.
\end{lemma}

\begin{proof}
Consider two simplices $\sigma,\sigma'\in\I$ satisfying $\sigma\subseteq\sigma'$, and let $\tau\in\Delta'(\sigma)$.
By definition, the local task $\Pi_{\tau,\sigma}=(\tau,\Delta(\sigma),\Delta_{\tau,\sigma})$ is solvable in one round,
by a simplicial map
${g:\Bary^{(1)}(\tau)\to\Delta(\sigma)}$ that agrees with $\Delta_{\tau,\sigma}$.
To show that $\tau\in\Delta'(\sigma')$, we show that the local task $\Pi_{\tau,\sigma'}=(\tau,\Delta(\sigma'),\Delta_{\tau,\sigma'})$ is solvable in one round,  using the same map $g$.
Since $\Delta(\sigma)\subseteq\Delta(\sigma')$, $g$ is a simplicial map from $\Bary^{(1)}(\tau)$ to $\Delta(\sigma')$.
To see that $g$ agrees with $\Delta_{\tau,\sigma'}$, let $\tau'\subseteq\tau$. If $\dim(\tau')> 0$, then
$
g(\tau')\in \Delta_{\tau,\sigma}(\tau')=\Delta(\sigma)\subseteq \Delta(\sigma')=\Delta_{\tau,\sigma'}(\tau').
$
If $\tau'=\{v\}$ is a vertex,
 then
$
g(v)\in \Delta_{\tau,\sigma}(v)=v= \Delta_{\tau,\sigma'}(v).
$
It follows that $g$ agrees with $\Delta'_{\tau,\sigma'}$, and thus $\tau\in\Delta'(\sigma')$, from which we conclude that ${\Delta'(\sigma)\subseteq \Delta'(\sigma')}$, i.e., $\Delta'$ is a carrier map.
\end{proof}

\subsubsection{Examples}
\label{examples: colorless closure}

\begin{itemize}
\item
Let $k\geq 3$.
For the set-agreement task $\SA_k$, the closure is solvable in zero rounds by having each process outputting its input,
since any combination of at most $k$ input values forms a valid output configuration 
of $\Cl(\SA_k)$.
To prove this, 
we show that for every $\sigma\in\I$
we have $\sigma\in\Delta'(\sigma)$.
First, note that for $\sigma\in\I$ such that $\sigma\neq [k]$, we have $\sigma\in\Delta(\sigma)$ and 
$\sigma\in\Delta(\sigma)\Rightarrow \sigma\in\Delta'(\sigma)$,
so we only have to show
$[k]\in\Delta'([k])$.
This is true since the local task $\Pi_{[k],[k]}$ is solvable in one round, by letting each process that sees more than one value output~$1$.

\item
On the other hand, for the Hexagon task $\HX$ we have $\Cl(\HX)=\HX$.
To see this, consider an input simplex $\sigma=(\{u_i,u_{i+1 \bmod 3}\})$, for some $i\in\{0,1,2\}$, its image $\Delta(\{u_i,u_{i+1 \bmod 3}\})=\{v_i,v_{i+1}\}\cup \{v_{i+3},v_{i+4\bmod 6}\}$,
and two vertices in its image that do not already constitute a simplex, i.e. $w\in\{v_i,v_{i+1}\}$ and $w'\in\{v_{i+3},v_{i+4\bmod 6}\}$.
Let $\tau=\{w,w'\}$.

If the local task $\Pi_{\tau,\sigma}=(\tau,\Delta(\sigma),\Delta_{\tau,\sigma})$ would have been solvable in one round,
then there would have been a map $f:V(\Bary^{(1)}(\tau))\to \Delta(\sigma)$ satisfying $f(\{w\})=w$ and $f(\{w'\})=w'$.
But $\Bary^{(1)}(\{w,w'\})$ is
\[
\{w\} \; \rule{3em}{2pt} \; \{w,w'\}  \; \rule{3em}{2pt} \; \{w'\}
\]
and any such map cannot be simplicial by continuity:
if $f(\{w,w'\})\in\{v_i,v_{i+1}\}$ then $\{f(\{w'\}),f(\{w,w'\})\}\notin\Delta(\sigma)$, and similarly
if $f(\{w,w'\})\in\{v_{i+3},v_{i+4\bmod 6}\}$ then $\{f(\{w\}),f(\{w,w'\})\}\notin\Delta(\sigma)$.
\end{itemize}

\subsection{Colorless Speedup Theorem}
\label{subsec:colorless speedup}

We now establish our speedup theorem for colorless algorithms, as stated next.

\begin{theorem}\label{theo:closure}
For every colorless task $\Pi=(\I,\O,\Delta)$, every $t>0$, and every $n\geq 2$, if $\Pi$ is solvable in $t$ rounds by $n$ processes running a wait-free colorless algorithm then the closure $\Cl(\Pi)$ is solvable in $t-1$ rounds by $n$ processes running a wait-free colorless algorithm.
\end{theorem}

\begin{proof}
Let  $\Pi=(\I,\O,\Delta)$, $t>0$, $n\geq 2$, and $\Cl(\Pi)=(\I,\O',\Delta')$ (see Fig.~\ref{fig:theo-weak}).
Since $\Pi$ is solvable in $t$ rounds by a colorless algorithm, there exists a simplicial map
\[
f:\Bary^{(t)}(\I) \to\O
\]
that agrees with~$\Delta$, i.e., for every $\sigma\in\I$ and  $\rho\in\Bary^{(t)}(\sigma)$, $f(\rho)\in\Delta(\sigma)$.
We show the existence of a simplicial map
\[
f':\Bary^{(t-1)}(\I)\to\O'
\]
that agrees with~$\Delta'$, which will establish the fact that $\Cl(\Pi)$ is solvable in $t-1$ rounds by a colorless algorithm. Specifically, for every $u\in \Bary^{(t-1)}(\I)$, we set
\[
f'(u)=f(\{u\}).
\]
That is, $f'(u)$ is the output value produced by $f$ at the vertex $w=\{u\}$ corresponding to the view~$w$ after $t$~rounds of a process with view $u$ after $t-1$~rounds, and running solo at round~$t$.

\begin{figure}[!t]
\centering
\includegraphics[width=12cm]{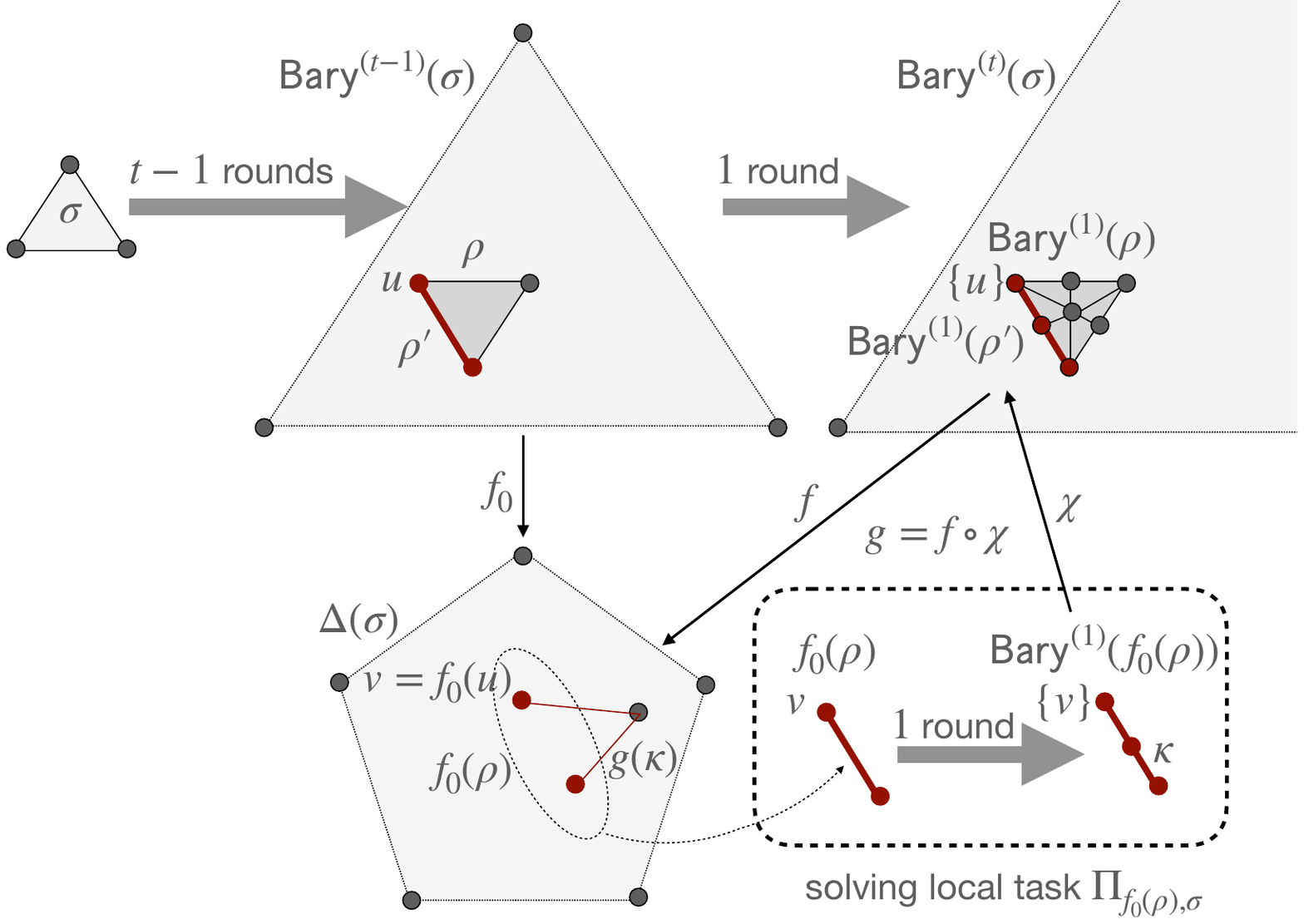}
\caption{\sl Construction in the proof of Theorem~\ref{theo:closure}}
\label{fig:theo-weak}
\end{figure}

To show that $f'$ is simplicial and agrees with~$\Delta'$, let $\sigma\in\I$ and let
\[
\rho\in\Bary^{(t-1)}(\sigma).
\]
We show that $f'(\rho)\in\Delta'(\sigma)$, by showing that the local task $\Pi_{f'(\rho),\sigma}$ is solvable in one round by a colorless algorithm. It is sufficient to show that there exists a simplicial map
\[
g:\Bary^{(1)}(f'(\rho))\to\Delta(\sigma)
\]
that agrees with $\Delta_{f'(\rho),\sigma}$. To show the existence of~$g$, we use the fact that $\Bary^{(1)}(f'(\rho))$ is isomorphic to  $\Bary^{(1)}(\rho')$ where $\rho'\subseteq\rho$ is any of the faces of $\rho$ for which $f':\rho'\to f(\rho)$ is a bijection. Indeed, $f':\rho\to f'(\rho)$ is onto by construction, but it may not be one-to-one. Let
\[
\chi:\Bary^{(1)}(f'(\rho))\to\Bary^{(1)}(\rho')
\]
be such an isomorphism, satisfying that, in addition, for every vertex~$u$ of~$\rho'$,
\[
\chi(\{f'(u)\})=\{u\}.
\]
That is, $\chi$~is consistent with respect to the solo executions starting from $u\in V(\rho')$ and from $f'(u)\in V(f'(\rho))$, for every vertex $u$ of~$\rho'$. As $\Bary^{(1)}(\rho')$ is a sub-complex of $\Bary^{(t)}(\sigma)$, the map
\[
g=f\circ \chi
\]
is well defined. Moreover, since $f$ is simplicial, $g$ is simplicial too. It remains to show that $g$ agrees with $\Delta_{\tau,\sigma}$. For this, let $\tau'\subseteq f'(\rho)$ be a face of~$f'(\rho)$. If $\tau'$ is a vertex~$v$, then let $u=f'^{-1}(v)$ be the pre-image of~$v$ in~$\rho'$. We have $\Bary^{(1)}(v)=\{v\}$. It follows that
\[
g(v)=f\circ\chi(\{v\})=f\circ\chi(\{f'(u)\})=f(\{u\})=f'(u)=v.
\]
Thus, $g(v)=v$, as desired. If $\dim(\tau')\geq 1$, then let $\kappa$ be a face of $\Bary^{(1)}(\tau')$. We have
\[
g(\kappa)=f\circ\chi(\kappa)=f(\chi(\kappa))\in \Delta(\sigma),
\]
merely because $\chi(\kappa)$ is a simplex of
\[
\Bary^{(1)}(\rho')\subseteq \Bary^{(1)}(\rho)\subseteq \Bary^{(1)}(\Bary^{(t-1)}(\sigma))= \Bary^{(t)}(\sigma),
\]
and $f$ agrees with~$\Delta$. Therefore, $g$~agrees with~$\Delta_{\tau,\sigma}$, and thus $f'(\rho)\in\Delta'(\sigma)$, from which it follows that $\Cl(\Pi)$ is solvable in $t-1$ rounds.
\end{proof}

Since the colorless closure task $\Cl(\Pi)=(\I,\O',\Delta')$ of a colorless task $\Pi=(\I,\O,\Delta)$ only potentially adds valid output simplices to $\Delta(\sigma)$ for forming $\Delta'(\sigma)$, for every $\sigma\in\I$, it follows  that, after applying the closure operator for some finite number of times~$t$, we get a fixed point, i.e. a task $\Cl^{(t)}(\Pi)$ such that $\Cl^{(t+1)}(\Pi)=\Cl^{(t)}(\Pi)$.
Naturally, this also implies  $\Cl^{(t')}(\Pi)=\Cl^{(t)}(\Pi)$ for every $t'\geq t$.

\begin{definition}
The \emph{fixed-point} of a colorless task $\Pi$ is the task $\Pi^\ast=(\I,\O^\ast,\Delta^\ast)$ such that $\Pi^\ast=\Cl^{(t)}(\Pi)$ for some $t\geq0$, and $\Cl^{(t+1)}(\Pi)=\Cl^{(t)}(\Pi)$.
\end{definition}

As a direct consequence of the speedup theorem (Theorem~\ref{theo:closure}) the fixed-point task $\Pi^*$ of a task~$\Pi$ is either $0$-round solvable, or not solvable at all. Indeed, consider $\Pi^\ast=(\I,\O^\ast,\Delta^\ast)$
	and assume it is solvable in $t>0$ rounds.
	By Theorem~\ref{theo:closure}, $\Pi^\ast=\Cl(\Pi^\ast)$ is solvable in $t-1$ rounds.
	Repeating this argument for $t$ times implies that $\Pi^\ast$ is $0$-round solvable.

\begin{lemma}
\label{lem:0round}
Let $n\geq 2$. Given a colorless task $\Pi=(\I,\O,\Delta)$,
its fixed-point $\Pi^\ast$ is either $0$-round solvable by $n$ processes running a wait-free colorless algorithm, or not solvable at all.
\end{lemma}

\subparagraph{Example.}
The following corollary of Theorem~\ref{theo:closure} illustrates both the theorem's interest and its simplicity.

\begin{corollary}
For every $n\geq2$, the hexagon task cannot be solved by $n$ processes running a wait-free colorless algorithm.
\end{corollary}

\begin{proof}
If $\HX$ is solvable wait-free by $n\geq2$ processes in the IIS model, then there exists $t\geq 0$ such that  $\HX$ is solvable in $t$ rounds.
We have seen that $\Cl(\HX)=\HX$. It follows from Theorem~\ref{theo:closure} that if $\HX$ is solvable wait-free, then it is solvable wait-free in zero rounds.

Consider a possible zero-round algorithm for $\HX$ with $n\geq2$ processes, and its decision map~$\delta$.
As the algorithm must produce valid outputs for executions with a single input, we have $\delta(u_i)\in\{v_i,v_{i+3}\}$ for every $i\in\{0,1,2\}$.
Let $v_j=\delta(u_0)$ (and hence $j\in\{0,3\}$).
The definition of $\Delta$ for executions with two different inputs guarantees $\delta(v_1)\in \{u_j,u_{j+1}\}$ and by the above we have
$\delta(v_1)=u_{j+1}$.
Similarly,
$\delta(v_2)=u_{j+2}$,
and hence
$\delta\{v_0,v_2\}=\{u_j,u_{j+2}\}\notin\Delta\{v_0,v_2\}$, a contradiction.
\end{proof}


\section{The Topology of the Closure}
\label{sec:connectivity of closure}

In this section, we study the topology of the colorless closure task, which, as opposed to the general closure, displays very simple properties. 
For stating these properties, let us recall that a complex $\K$ is \emph{complete} if $V(\K)$ is a simplex of~$\K$.
It is \emph{complete up to dimension~$d$} if every set $\tau\subseteq V(\K)$ with $0\leq \dim(\tau)\leq d$ satisfies $\tau\in\K$.
A \emph{connected component} of a complex $\K$ is the subcomplex of $\K$ induced by all the vertices in a connected component of the graph $\Skel_1(\K)$. We mainly show that, for every $\sigma\in\I$, the closure of each connected component of $\Delta(\sigma)$ eventually becomes complete after a bounded number of closure operations. 


For a colorless task $\Pi=(\I,\O,\Delta)$, we
denote by  $\Cl^{(k)}(\Pi)=(\I,\O^{(k)},\Delta^{(k)})$
the $k$-th closure of~$\Pi$, 
defined  as
$\Cl^{(k)}(\Pi)=\Cl(\Cl^{(k-1)}(\Pi))$,
where $k$ is a positive integer and $\Cl^{(0)}(\Pi)=\Pi$.

\begin{theorem}\label{theo:connectivity}
Given a colorless task $\Pi=(\I,\O,\Delta)$ and a simplex $\sigma\in \I$,
the following hold.
\begin{itemize}
\item
Let $D$ be the largest diameter of a connected component in the graph $\Skel_1(\Delta(\sigma))$, and let $\ell=\lceil\log_2D\rceil+1$.
Then all the connected components of $\Delta^{(\ell)}(\sigma)$ are complete up to dimension~$\dim(\sigma)$.
\item For every $k\geq 0$, a set of vertices in $V(\Delta(\sigma))$ is a connected component of
$\Skel_1(\Delta(\sigma))$ if and only if it is a connected component of $\Skel_1(\Delta^{(k)}(\sigma))$.
\end{itemize}
\end{theorem}

\begin{proof}
To establish the theorem, we first prove two auxiliary claims, with the same notations as in the statement of the theorem.

\begin{claim}\label{claim:diameter}
	Every two vertices $u\neq w$ of the same connected component of $\Delta(\sigma)$ satisfy $\{u,w\}\in\Delta^{(\ell-1)}(\sigma)$.
\end{claim}

\noindent\emph{Proof of claim.}
Let $u,v,w$ be three vertices of the same connected component~$\K$, such that $\{u,v\}\in \K, \{v,w\}\in \K$, but $\{u,w\}\notin \K$, and let us show that $\{u,w\}\in \Delta^{(1)}(\sigma)$.
(If there are no such three vertices then we are done.)
For this, it is sufficient to define a  simplicial map
\[
f:\Bary(\{u,w\})\to \Delta(\sigma)
\]
which agrees with $\Delta_{\{u,w\},\sigma}$. We set
\[
f(\{u\})=u, \; f(\{w\})=w, \; \mbox{and} \; f(\{u,w\})=v.
\]
In this way, any edge of $\Bary(\{u,w\})$ is mapped to either  $\{u,v\}$ or $\{v,w\}$, which both belong to~$\Delta(\sigma)$. It follows that $f$ is simplicial, and agrees with  $\Delta_{\{u,w\},\sigma}$.
Therefore, for any two vertices $u,w$ of $\K$ at distance at most~2 in the graph $\Skel_1(\K)$, $\{u,w\}\in \Delta^{(1)}(\sigma)$.
By the same argument, any two vertices $u,w$ of $\K$ at distance at most~4 in the graph $\Skel_1(\K)$ satisfy $\{u,w\}\in \Delta^{(2)}(\sigma)$, and more generally, for any two vertices $u,w$ of $\K$ at distance at most~$2^r$ in the graph $\Skel_1(\K)$, $\{u,w\}\in \Delta^{(r)}(\sigma)$.
As a consequence, for every two vertices $u,w$ of~$\K$, $\{u,w\}\in \Delta^{(\ell-1)}(\sigma)$.
\hfill $\diamond$

\bigskip 

We next show that a similar claim holds for any set of vertices in a connected component of $\Delta(\sigma)$, and not only for pairs.

\begin{claim}\label{claim:therest}
	Every set  $\tau\subseteq V(\Delta(\sigma))$ of vertices of the same connected component of $\Delta(\sigma)$
	with $2<  |\tau| \leq  |\sigma|$
	satisfies
	$\tau\in \Delta^{(\ell)}(\sigma)$.
\end{claim}

\noindent\emph{Proof of claim.} It is sufficient to show that the local task $\Pi^{(\ell-1)}_{\tau,\sigma}=(\tau,\Delta^{(\ell-1)}(\sigma),\Delta^{(\ell-1)}_{\tau,\sigma})$ is solvable in one round, which we do by defining a simplicial map
\[
f:\Bary(\tau)\to \Delta^{(\ell-1)}(\sigma)
\]
that agrees with $\Delta^{(\ell-1)}_{\tau,\sigma}$,
as follows.
Let $\tau=\{v_1,\dots,v_d\}$ where the vertices of~$\tau$ are indexed in an arbitrary order, where $d=|\tau|$.
For every singleton set~$\{v_i\}\subseteq \tau$ we set 
\[
f(\{v_i\})=v_i
\]
while for every set $S \subseteq \tau$ of cardinality $|S|>1$, we set 
\[
f(S)=v_1.
\]
Let $\tau'$ be a face of $\tau$.
The views encoded by different vertices in a simplex $\rho\in \Bary(\tau')$ of a barycentric subdivision are totally ordered by inclusion, so $\rho$ may contain at most one vertex that corresponds to singleton view (i.e., a view composed of a single vertex of~$\tau'$).
As a consequence, there are only three possible cases:
\begin{itemize}
	\item $f(\rho)=v_i$ whenever $\rho=\{v_i\}$, or
	\item $f(\rho)=v_1$ whenever $\rho$ contains no singleton sets but possibly $\{v_1\}$, or
	\item $f(\rho)=\{v_1,v_i\}$ for some $i\in\{2,\dots,d\}$ whenever $\rho$ contains the singleton~$\{v_i\}$ plus other non-singleton vertices.
\end{itemize}
In all three cases, the image of $\rho$ is a simplex of $\Delta^{(\ell-1)}_{\tau,\sigma}(\tau')$,
by Claim~\ref{claim:diameter}. Therefore $f$ is simplicial, and agrees with~$\Delta^{(\ell-1)}_{\tau,\sigma}$, which implies that $\tau\in \Delta^{(\ell)}(\sigma)$.
 \hfill $\diamond$

\bigskip 

We are now ready to prove Theorem~\ref{theo:connectivity}. For the first part of the theorem, let $\K$ be a connected component of~$\Delta(\sigma)$, and our goal is to show that every $\tau\subseteq V(\K)$  with $1\leq |\tau|\leq \dim(\sigma)+1$ satisfies $\tau\in\Delta^{(\ell)}(\sigma)$.
The claim holds for $|\tau|=1$ (i.e., for vertices) as every vertex of~$\K$ is by definition a vertex of $\Delta(\sigma)$.
It holds for $|\tau|=2$ (i.e., for edges) by Claim~\ref{claim:diameter}, and for $|\tau|>2$
by Claim~\ref{claim:therest}.

For establishing the second item, first note that one direction of the if and only if statement is trivial, as the closure operator can only add simplices, so connected components of $\Skel_1(\Delta^{(k-1)}(\sigma))$ can only merge when moving to  $\Skel_1(\Delta^{(k)}(\sigma))$ and not break.
For the other direction, we proceed by induction on $k\geq 0$.
The statement is trivial for $k=0$.

Let us assume that the statement holds for~$k$ and show that it holds for~$k+1$.
Let $u$ and $v$ be two vertices of $\Delta(\sigma)$ in two different connected components of  $\Delta(\sigma)$, and by assumption
in two different connected components of $\Delta^{(k)}(\sigma)$.
Let us show that $\{u,v\}\notin \Delta^{(k+1)}(\sigma)$.
For the purpose of contradiction, let us consider a map
 $$f:\Bary(\{u,v\})\to \Delta^{(k)}(\sigma)$$
that agrees with $\Delta^{(k)}_{\{u,v\},\sigma}$. We must have $f(\{u\})=u$, $f(\{v\})=v$, and $f(\{u,v\})=w$ for some vertex $w\in V(\Delta^{(k)}(\sigma)$). However, $u$ and $v$ are  in two different connected components of $\Delta^{(k)}(\sigma)$, which implies that $\{u,w\}$ or $\{v,w\}$ is not an edge of $\Delta^{(k)}(\sigma)$.
As a consequence, $f$ is not simplicial, and thus $\{u,v\}\notin \Delta^{(k+1)}(\sigma)$. In other words, no edges can be added between different connected components of $\Delta(\sigma)$ during successive closure operations.
\end{proof}

\paragraph*{Examples}
 \begin{itemize}
	\item  
	In Section~\ref{examples: colorless closure} we have proved that the closure of $k$-set agreement contains any combination of at most $k$ input values.
	This fact can now be directly derived from Theorem~\ref{theo:connectivity},
	as $\Delta([k])$ is connected and contains all the values of $[k]$.
	
	\item 
	In the same section, we have proved that $\Cl(\HX)=\HX$.
	This is also a direct consequence of Theorem~\ref{theo:connectivity},
	as for every $\sigma\in\I$, each of the connected components of $\Delta(\sigma)$ is full.
\end{itemize}


\section{Round-Reduction is Complete for 1-Dimensional Tasks}

1-dimensional tasks are tasks for which the input and output complexes are of dimension at most~1. 
In this section, we establish the completeness of the round-reduction proof technique for 1-dimensional colorless tasks, thus establishing part of Theorem~\ref{theo:C}.
FLP-style proofs are also complete in this case, but the proof of this fact is deferred to the next section, where we show that the techniques are equivalent.
The next theorem asserts that for colorless tasks of dimension at most~$1$
the reciprocal of Theorem~\ref{theo:closure} holds as well, thus establishing the completeness of the round-reduction technique in this case.

\begin{theorem}\label{theo:closure-reciprocal}
For every $1$-dimensional colorless task $\Pi=(\I,\O,\Delta)$, every $t>0$, and every $n\geq 2$, if the colorless closure $\Cl(\Pi)$ is solvable in $t-1$ rounds by $n$ processes running a wait-free colorless algorithm, then $\Pi$ is solvable in $t$ rounds by $n$ processes running a wait-free colorless algorithm.  
\end{theorem}

\begin{proof}
Let $\Pi=(\I,\O,\Delta)$ such that $\Cl(\Pi)$ is solvable in $t-1$ rounds by $n$ processes. Let 
\[
\delta':\Bary^{(t-1)}(\Skel_n(\I))\to \O
\]
be a simplicial map agreeing with~$\Delta$. Roughly, an algorithm for solving $\Pi$ consists of two phases: first solving $\Cl(\Pi)$ using~$\delta'$, and, second, given the output of~$\delta'$, reconciliating these outputs (valid for $\Cl(\Pi)$, but not necessarily for~$\Pi$) using the algorithm solving the local task for these outputs (see Fig.~\ref{fig:theo-reciprocal}). More formally, les us consider a process~$p$ with input~$x$, i.e., $x\in \I$ is a vertex. After $t-1$ rounds, this process gets a view $w$ which is a vertex of $\Bary^{(t-1)}(\I)$. Then, $p$ proceeds with one more round of communication, and gets a view in $\Bary^{(t-1)}(\I)$, 
which is either of the form $\{w\}$ or of the form $\{w,w'\}$ where $\{w,w'\}$ is an edge of $\Bary^{(t-1)}(\I)$. 
Note that the property of the Barycentric subdivision guarantee that in the latter case $w'$ must contain an input $x'\neq x$ of another process, where $e=\{x,x'\}$ was the actual input. 
Let 
$
y=\delta'(w)\;\mbox{and}\; y'=\delta'(w').
$
Moreover, let $e'=\{y,y'\}$. Note that $e'$ is an edge of $\Delta'(e)$ as $\delta'$ solves $\Cl(\Pi)$, but it is not necessarily an edge of~$\Delta(e)$. The algorithm solving~$\Pi$ is as follows: 
\begin{itemize}
\item If the view of~$p$ after $t$ rounds is $\{w\}$, then $p$ outputs~$y$; 
\item If the view of~$p$ after $t$ rounds is $\{w,w'\}$, then $p$ outputs~$z=f_{e',e}(\{y,y'\})$ where 
\[
f_{e',e}:\Bary^{(1)}(e')\to \Delta(e)
\] 
is a simplicial map\footnote{There might be more than one simplicial map from $\Bary^{(1)}(e')$ to $\Delta(e)$, in which case one selects one of them arbitrarily for defining the algorithm solving~$\Pi$.} solving the local task $\Pi_{e',e}$. 
\end{itemize}
This algorithm is well defined, as if $p$ has view $\{w,w'\}$, then it knows $x$ and $x'$, and it can compute $y$ and $y'$ from the views $w$ and~$w'$. 
To show correctness, let $\sigma\in \I$, and let us show that our algorithm produces a simplex $\tau\in \Delta(\sigma)$. 
If $\sigma$ is a vertex~$x$, then all processes output $y\in\Delta'(x)$, which is a vertex of~$\Delta(x)$, as desired. 
If $\sigma$ is an edge $e=\{x,x'\}$, then let $\{w,w'\}\in\Bary^{(t-1)}(e)$ be the edge of the barycentric subdivision corresponding the the current configuration after $t-1$ rounds. Assume, w.l.o.g., that, during the $t$-th round, some processes (maybe none) get view~$\{w\}$ while some other processes (at least one) get view $\{w,w'\}$ --- the latter view is a vertex of $\Bary^{(t)}(e)$, whereas it was an edge of $\Bary^{(t-1)}(e)$. 
Note that starting from $\{w,w'\}\in\Bary^{(t-1)}(e)$, it is not possible that some processes reads a view $\{w\}$ in $\Bary^{(t)}(e)$ while some other processes reads $\{w'\}$ in $\Bary^{(t)}(e)$. 
It follows that a group of processes may output $y=\delta'(x)$ while another group of processes may output $z=f_{e',e}(\{y,y'\})$. The crucial property is that $f_{e',e}$ fixes vertices, that is, $f_{e',e}(\{y\})=y$. Since $f_{e',e}$ is simplicial and agrees with $\Delta_{e',e}$, we get that 
\[
\{y,z\}=\{f_{e',e}(\{y\}),f_{e',e}(\{y,y'\})\}\in \Delta_{e',e}(e')=\Delta(e), 
\]
as desired. This completes the proof of the theorem. 
\end{proof}

\begin{figure}[!t]
\centering
\includegraphics[width=12cm]{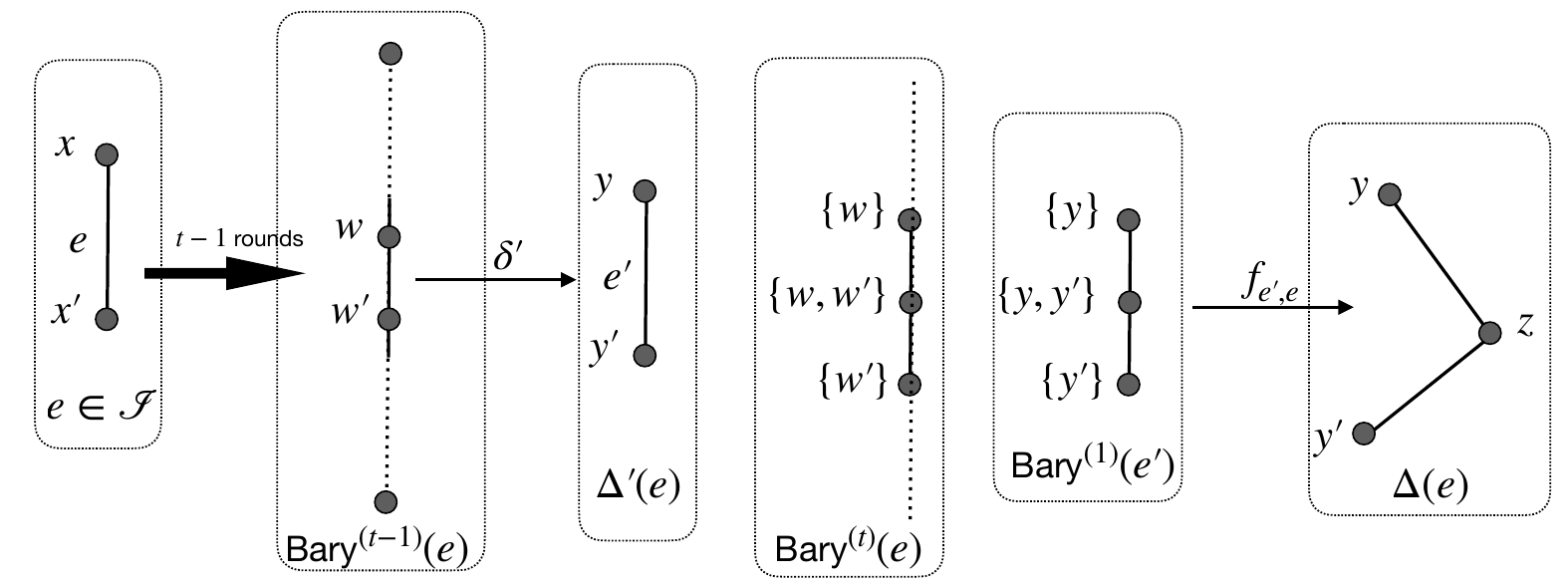}
\caption{\sl Proof of Theorem~\ref{theo:closure-reciprocal}}
\label{fig:theo-reciprocal}
\end{figure}

By combining Theorems~\ref{theo:closure} and~\ref{theo:closure-reciprocal}, we obtain the desired result. 

\begin{corollary}\label{cor:round-reduction-complete}
The round-reduction proof techniques is complete for 1-dimensional colorless tasks and wait-free colorless algorithms.
\end{corollary}


\section{Relations Between Round-Reduction and FLP-Style Proofs}
\label{sec:reductions}

In this section we establish tight connections between the FLP-style proof strategy and the round-reduction proof technique, in the context of wait-free solvability of colorless tasks.
Specifically, we first establish Theorem~\ref{theo:B} which asserts
	that the existence of a round-reduction proof implies the existence of an FLP-style proof,
from which the remaining part of Theorem~\ref{theo:C} (completeness of FLP-style proofs for 1-dimensional tasks) will follow.
This implies that the
round-reduction techniques and FLP-style proofs have the same power when considering impossibility proofs for 1-dimensional colorless tasks.
In Theorem~\ref{theo:ex-to-rr-colorless-dim1},
we give a direct proof for this:
the existence of
an FLP-style proof implies the existence of a round-reduction proof for the impossibility of such tasks
(complementing Theorem~\ref{theo:B}, for 1-dimensional tasks).
This direct proof may suggest that a more tight connection between the two proof techniques exists.
We start by formally defining FLP-style proofs.

\subsection{FLP-Style Proofs}

Given a colorless task $\Pi=(\I,\O,\Delta)$, an FLP-style proof constructs an infinite sequence of simplices $\sigma_0,\sigma_1,\dots$, where $\sigma_0\in \I=\Bary^{(0)}(\I)$, and, for every $t\geq 1$, $\sigma_t\in \Bary^{(1)}(\sigma_{t-1})\subseteq \Bary^{(t)}(\I)$, as follows.

The proof $P$ assumes for contradiction the existence of an algorithm $A$ solving $\Pi$.
$P$ starts by asking $A$ to reveal, for every $\sigma\in\I$, the \textit{valency} of~$\sigma$, that is, the set of output values that are returned by $A$ in executions starting with the input values forming the simplex~$\sigma$.
Then, $P$ chooses a simplex $\sigma_0\in \I$; in order to create an infinite sequence, $P$ must choose $\sigma_0$ such that $A$ is not able to claim it terminates in $0$ rounds, i.e., such that any possible assignment of outputs to the processes in~$\sigma_0$ is inconsistent with the valencies of~$\I$.

Given a sequence $\sigma_0,\dots,\sigma_t$ constructed by $P$ so far, $\sigma_{t+1}$ is obtained analogously, as follows. $P$ asks $A$ to reveal the valencies of all simplices in $\Bary^{(1)}(\sigma_t)$, that is,
for each $\sigma\in\Bary^{(1)}(\sigma_t)$, the set of output values produced by $A$ in all valid executions starting from~$\sigma$.
Based on these valencies, $P$ chooses one simplex $\sigma_{t+1}\in\Bary^{(1)}(\sigma_t)$;
as in the choice of $\sigma_0$, it must choose $\sigma_{t+1}$ such that $A$ is not able to assign outputs to the processes in $\sigma_{t+1}$ consistent with the valencies of~$\Bary^{(1)}(\sigma_t)$.

Let $\val:2^{\sigma_0}\cup 2^{\sigma_1}\cup\cdots\to 2^\O$ be the function defined by the valencies returned by $A$.
For any correct algorithm, the function $\val$ must satisfy some basic conditions of consistency with the task specification and with other valencies returned.
We next formalize these conditions for  $\sigma_t$.
\begin{itemize}
	\item
	Consistent with itself:
	each $\sigma\in\Bary^{(1)}(\sigma_t)$ satisfies
	$\val(\sigma)\subseteq\val(\sigma_t)$;
	moreover, $\cup_{\sigma\in\Bary^{(1)}(\sigma_t)}\val(\sigma)=\val(\sigma_t)$.	
	
	\item
	Consistent with $\Delta$:
	For each $\sigma\in\Bary^{(1)}(\sigma_t)$, $\val(\sigma)\subseteq\Delta(\sigma_0)$.

	\item
	Monotone: for each $\sigma'\subseteq\sigma\in\Bary^{(1)}(\sigma_t)$,
	$\val(\sigma')\subseteq\val(\sigma)$.
\end{itemize}

If this strategy can proceed forever, constructing an infinite sequence $\sigma_0,\sigma_1,\dots$ of simplices, then $A$ does not terminate in this execution, disproving the existence of an algorithm~$A$ solving $\Pi$.
At the core of FLP-style proofs stands a \emph{choice mechanism} that picks the next simplex~$\sigma_{t+1}$.
We present such a mechanism for one-dimensional colorless tasks.

Our work on the FLP-style technique continues recent lines of work regarding the power of extension-based proofs~\cite{AlistarhAEGZ19,AttiyaCR20,BrusseE21}.
We allow less diverse queries compared to these works, yet our results imply that if a one-dimensional colorless task is unsolvable, then the simple queries we allow are sufficient for proving this impossibility.

\paragraph*{Example}
The impossibility of solving the hexagon task can be proved using an FLP-style proof,
by constructing a sequence $\sigma_0,\sigma_1,\ldots$ as follows.
All the simplices $\sigma_0,\sigma_1,\ldots$ will be edges, and we will fix
$i\in\{0,1,2\}$
such that each $\sigma_t$ satisfies the invariants
\[(1)\;\;
 \val(\sigma_t)\subseteq\{v_i,v_{i+1}\}\cup\{v_{i+3},v_{i+4}\}\text{, and}\]
\[(2)\;\;
 \val(\sigma_t)\cap\{v_i,v_{i+1}\}\neq\emptyset \text{ and } \val(\sigma_t)\cap\{v_{i+3},v_{i+4}\}\neq\emptyset,\]
where the indices here and below are computed modulo $6$, unless otherwise specified.

To choose $\sigma_0$ and $i$, inspect the valencies of the three edges $(u_i,u_{i+1\bmod 3})$, for $i\in\{0,1,2\}$.
Since $\val(u_i)\subseteq \val (u_{i-1\bmod 3},u_i)\cap \val(u_i,u_{i+1\bmod 3})$ and valency is monotone, the valencies of every two edges must intersect, and thus for some $i\in\{0,1,2\}$ the edge $e=(u_i,u_{i+1\bmod 3})$
must satisfy both
$\val(e)\cap\{v_i,v_{i+1}\}\neq\emptyset$ and $\val(e)\cap\{v_{i+3},v_{i+4}\}\neq\emptyset$;
we fix this $i$, and set $\sigma_0=e$,
guaranteeing Invariant~$(2)$.
Invariant~$(1)$ holds since the valency must be consistent with~$\Delta$.

Assume $\sigma_0,\ldots,\sigma_t$ are chosen and satisfy both invariants, and that $\val(\sigma)$ is known for each $\sigma\in\Bary_1(\sigma_t)$.
Since $\val(\sigma)\subseteq\val(\sigma_t)$ for every
$\sigma\in\Bary_1(\sigma_t)$,
Invariant~$(1)$ will hold for any
$\sigma_{t+1}\in\Bary_1(\sigma_t)$ we may choose.
Let $\sigma_t=\{w_0,w_1\}$,
then $\Bary_1(\sigma_t)$ is composed of the vertices $\{w_0\},\{w_0,w_1\},\{w_1\}$ and the edges
$e_0=(\{w_0\},\{w_0,w_1\})$ and
$e_1=(\{w_1\},\{w_0,w_1\})$.

We have $\val(e_0)\cap\val(e_1)\neq\emptyset$,
and by Invariant~$(1)$ we also have $\val(e_0)\cap\val(e_1)\subseteq
\val(\sigma_t)\subseteq\{v_i,v_{i+1},v_{i+3},v_{i+4}\}$,
so at least one value $v\in\{v_i,v_{i+1},v_{i+3},v_{i+4}\}$ satisfies
$v\in\val(e_0)\cap\val(e_1)$;
assume without loss of generality that $v\in\{v_i,v_{i+1}\}$.

The valencies of the vertices are contained in the valencies of the edges,
hence $\val(e_0)\cup\val(e_1)=\val(\sigma_t)$,
so Invariant~$(2)$ implies that
both
$(\val(e_0)\cup\val(e_1))\cap\{v_i,v_{i+1}\}\neq\emptyset$ and
$(\val(e_0)\cup\val(e_1))\cap\{v_{i+3},v_{i+4}\}\neq\emptyset$ hold.
Hence, at least one edge $e\in\{e_0,e_1\}$
has
$\val(e)\cap\{v_{i+3},v_{i+4}\}\neq\emptyset$,
and since $v\in\val(e)$ it also have
$\val(e)\cap\{v_{i},v_{i+1}\}\neq\emptyset$.
This edge is set as $\sigma_{t+1}$, and Invariant~$(2)$ is satisfied for $\sigma_{t+1}$ as well.
Since this process can continue for every $t\geq 0$, the proof is complete.

\subsection{Round-Reduction Proofs Imply FLP-Style Proofs}

This section is dedicated to the proof of Theorem~\ref{theo:B}.

\begin{theorem}
	\label{theo:rr-to-ex-colorless}	
	For every colorless task~$\Pi$ and $n\geq2$, if there is a round-reduction  proof establishing the impossibility of solving $\Pi$ by $n$ processes running a wait-free colorless algorithm, then there is an FLP-style proof establishing the same impossibility.
\end{theorem}

\begin{proof}	
	Fix a colorless task $\Pi$ that has a round-reduction impossibility proof.
	By Lemma~\ref{lem:0round}, $\Pi^\ast$ is not $0$-round solvable.	
	For each simplex $\sigma\in\O$,
	all the connected components of $\Delta^{\ast}(\sigma)$ are complete up to dimension~$\dim(\sigma)$ by Claim~\ref{claim:therest}:
	if some simplices of dimension at most $\dim(\sigma)$ are missing in $\Delta^{\ast}(\sigma)$ then at least one of them would have been added to it when applying the closure operator, contradicting the fact that $\Pi^\ast$ is a fixed-point.
	
	To construct an FLP-style proof, we consider an algorithm $A$ that claims to solve $\Pi$,
	and
	define $\delta:V(\I)\to V(\O)$ to map each input value $x\in\I$ to the output value $\delta(x)$ produced by~$A$ in the execution where only the value $x$ appears;
	$\delta(x)$ is unique since the execution is unique and the algorithm is deterministic.
	The fact that $\Cl^\ast(\Pi)$ is not $0$-round solvable means that
	there is a simplex $\sigma\in \I$ such that $\delta(\sigma)=\{\delta(x)\mid x\in\sigma\}\notin\Delta^{\ast}(\sigma)$.
	As the connected components are complete up to dimension $\dim(\sigma)$, the simplex $\sigma$ is mapped to (at least) two different connected components, i.e. there are two input values $x,x'\in \sigma$
	and two connected components $C,C'$ of $\Delta(\sigma)$ such that $\delta(x)\in C$ and $\delta(x')\in C'$.
	
	Let $\sigma_0=\{x,x'\}$.
	As $\sigma_0\subseteq\sigma$, the fact that $\Delta^\ast$ is a carrier map (Lemma~\ref{lem: closure is carrier})
	implies that the connected components of $\Delta(\sigma_0)$ are a refinement of the connected components of $\Delta(\sigma)$.
	Hence, there is
	a connected component $C_0\subseteq C$ of $\Delta(\sigma_0)$ such that 	
	$\val(x)\cap C_0\neq\emptyset$,
	and similarly
	a different connected component $C'_0\subseteq C'$ of $\Delta(\sigma_0)$ such that 	
	$\val(x')\cap C'_0\neq\emptyset$.
	Note that Theorem~\ref{theo:connectivity} asserts that the connected components of $\Delta(\sigma_0)$ and of $\Delta^{\ast}(\sigma_0)$ are the same.

	Let us say that a configuration reachable from $\sigma_0$ is \emph{bivalent} (w.r.t.~$\sigma_0$) if a valency query on it returns output values in at least two different connected components of $\Delta(\sigma_0)$.
	Note that $\sigma_0$ is bivalent by construction, and that if the algorithm is in a bivalent configuration it cannot decide without taking further steps.
	
	We construct an infinite sequence $\sigma_0,\sigma_1,\ldots$ of bivalent configurations.	
	Each $\sigma_t$ will consist only of two views
	\[
	\sigma_t=\{w_t,w'_t\}.
	\]
	
	Assume $\sigma_0,\ldots,\sigma_t$ are chosen and $\val(\sigma)$ is known for each $\sigma\in\Bary_1(\sigma_t)$.
	%
	Recall that $\Bary_1(\sigma_t)$ is composed of the vertices $\{w_t\},\{w_t,w'_t\},\{w'_t\}$ and the edges
	$e=(\{w_t\},\{w_t,w'_t\})$ and
	$e'=(\{w'_t\},\{w_t,w'_t\})$.
	Let $C_t$ be a connected component of $\Delta(\sigma_0)$ such that $\val(\{w_t,w'_t\})\cap C_t\neq\emptyset$.
	
	The fact that $\sigma_t$ is bivalent implies that there is a connected component
	$C'_t$ of $\Delta(\sigma_0)$, $C'_t\neq C_t$,
	such that $\val(\sigma_t)\cap C'_t\neq\emptyset$.
	As the valencies of the vertices are contained in the valencies of the edges,
	we have $\val(e)\cup\val(e')=\val(\sigma_t)$.
	Hence, at least one edge $f\in\{e,e'\}$
	has
	$\val(f)\cap C'_t\neq\emptyset$.
	Since $\{w_t,w'_t\}\in f$,
	it also has
	$\val(f)\cap C_t \neq\emptyset$.
	The edge $f$ is thus bivalent, and we set it as $\sigma_{t+1}$.
	This process can continue for every $t\geq 0$, and the proof is complete.
\end{proof}

\subsection{FLP-Style Proofs are Complete for 1-Dimensional Tasks}

This section is dedicated to the proof of Theorem~\ref{theo:C} for FLP-style proofs,
	i.e. for showing these are complete for 1-dimensional colorless tasks.
For this, first observe that if there is an FLP-style proof for the impossibility of a $1$-dimensional colorless task~$\Pi$, then $\Pi$ is unsolvable.
By the completeness of the round-reduction proofs (cf. Corollary~\ref{cor:round-reduction-complete}), there is a round-reduction impossibility proof for that task.
On the other hand, Theorem~\ref{theo:rr-to-ex-colorless} asserts that if a colorless task has a round-reduction impossibility proof then it also has an FLP-style impossibility proof.
The establishes the desired equivalence between the two forms of proofs for colorless tasks.

\begin{corollary}
	\label{cor:flp iff rr dim 1}
Let $\Pi$ be a $1$-dimensional colorless task.
There is an FLP-style proof for the impossibility of solving $\Pi$ using wait-free colorless algorithms if and only if there is a round-reduction proof of this impossibility for~$\Pi$.
\end{corollary}

This corollary can also be proved directly: one direction is  Theorem~\ref{theo:rr-to-ex-colorless} (for any dimension),
and the other is given in Theorem~\ref{theo:ex-to-rr-colorless-dim1} (below).
We get that, for 1-dimensional tasks, the FLP-style proof style can be mechanized, i.e., for any $1$-dimensional task~$\Pi$, the FLP-style proof technique succeeds for~$\Pi$ if and only if~$\Pi$ is not wait-free solvable in the IIS model.

\begin{corollary}
	\label{cor:flp is complete for 1 dim}
The FLP-style proof technique is complete for $1$-dimensional tasks.
\end{corollary}

\begin{proof}
By Theorem~\ref{theo:rr-to-ex-colorless}, if  the FLP-style proof technique fails for~$\Pi$, then the round-reduction proof also fails for~$\Pi$ as well.
By Lemma~\ref{lem:0round}, this implies that $\Pi^\ast$ is $0$-round solvable.
By Theorem~\ref{theo:closure-reciprocal}, the original task~$\Pi$ is solvable.
\end{proof}


We also give a direct proof for the converse of 
Theorem~\ref{theo:rr-to-ex-colorless} for 1-dimensional tasks.

\begin{theorem}
	\label{theo:ex-to-rr-colorless-dim1}
	For every $1$-dimensional colorless task~$\Pi$ and $n\geq2$, if there is an FLP-style proof for the impossibility of solving $\Pi$ by $n$ processes running a wait-free colorless algorithm, then there is a round-reduction proof for the same impossibility.
\end{theorem}

\begin{proof}
	Let $\Pi=(\I,\O,\Delta)$ be a task with $\dim(\I)\leq 1$, for which an FLP-style proof of impossibility exists.
	In order to prove that the impossibility of $\Pi$ also has a round-reduction proof, we first state and prove two claims.
	
	\begin{claim}
		\label{claim:eb-implies-connectivity}
		For every function $\val:\I\to 2^{\O}$ there exists an edge $e\in \I$ such that $\val(e)$ is not a connected subgraph of $\Delta(e)$.
	\end{claim}
	
	\begin{proof}[Proof of Claim~\ref{claim:eb-implies-connectivity}]
		Fix some  $\val:\I\to 2^{\O}$ and let $e=\{x,x'\}\in\I$ be the first simplex chosen by the FLP-style proof according to $\val$.
		Assume for contradiction that $\val(e)$ is connected in $\Delta(e)$,
		and let $y=\delta(x)$ and $y'=\delta(x')$.
		Let $(y=z_0,z_1,\ldots,z_k=y')$ be a path in $\delta(e)$ that cover all the nodes of $\delta(e)$, possibly with repetitions.
		If the length $k$ of the path is not a power of $2$, concatenate the last node to the path to complete its length to a such a power;
		let $t=\log_2k$.
		
		We show how to answer all valency queries after choosing $e$.
		Consider all $t$-round executions starting from $e$, and the corresponding $1$-dimensional subdivided simplex.
		This simplex is a path of $k$ edges, with the first node being $x$ (formally, the image of $x$ under the subdivision operator), then $k-1$ new nodes, and lastly $x'$.
		To the execution represented by the $i$-th node in the subdivision, assign the output value $z_i$.
		This assignment is consistent with $\val$ since $y$ is assigned to the $x$-solo execution, $y'$ to the $x'$-solo execution, each value in $\val(e)$ to some execution starting from $e$, and it maps every two neighboring execution to two neighboring output values.
		
		Using this assignment, the protocol can consistently answer any output or valency query starting from $e$, and the prover lose.
		Hence, no FLP-style proof exists, a contradiction. This completes the proof of Claim~\ref{claim:eb-implies-connectivity}.
	\end{proof}
	Note that in the above proof, we used the fact that we can map a $1$-dimensional simplex to a $0$-dimensional one (when concatenating copies of $z_k$ to the path);
	this would have not been possible with colored tasks.
	
	\begin{claim}\label{claim:zesecond}
		\label{claim:closure-1dim-unsolvable}
		The closure task $\Pi^\ast$ of $\Pi$ is not $0$-round solvable.
	\end{claim}
	
	\begin{proof}[Proof of Claim~\ref{claim:zesecond}]
		Assume for contradiction that $\Pi^\ast$ is $0$-round solvable, i.e. there is a simplicial map  $\delta:\I\to\O^\ast$.
		As $\delta$ is simplicial, the vertices of a simplex $e=\{x,x'\}\in \I$ are mapped to neighboring vertices $\delta(x),\delta(x')$ in $\O^\ast$, and we denote the connected component of $\Delta^\ast$ containing these vertices by $K_e$.
		Theorem~\ref{theo:connectivity} guarantees that $K_e$ is also the connected component of $\Delta(e)$ containing $\delta(x)$ and $\delta(x')$.
		
		We use $\delta$ to define a valency assignment
		$\val:\I\to 2^{V(\O)}$ by $\val(x)=\delta(x)$ for $x\in V(\I)$ and $\val(e)=V(K_e)$ when $|e|=2$,
		which is a legal valency assignment since  $x\in e$ implies $\val(x)\subseteq\val(e)$.
		However, for each $e\in\I$ we have that $\val(e)$ is connected in $\Delta(e)$,
		contradicting 	Claim~\ref{claim:eb-implies-connectivity}. This completes the proof of Claim~\ref{claim:zesecond}.
	\end{proof}
	
	The proof of Theorem~\ref{theo:ex-to-rr-colorless-dim1}
	is now immediate.
	The closure $\Pi^\ast$ of $\Pi$
	is not $0$-round solvable by Claim~\ref{claim:closure-1dim-unsolvable}.
	If $\Pi$ was solvable then a repeated use of Theorem~\ref{theo:closure} would have implied that $\Pi^\ast$ is $0$-round solvable, a contradiction.
\end{proof}


\section{Applications}
\label{sec:applications}

Throughout this paper, we have shown how the theory we developed applies for set agreement and the Hexagon task.
We complete the paper by presenting some further applications:
a round-reduction lower bound proof for approximate agreement, and
a round-reduction impossibility proof for covering tasks.
Note that in terms of techniques, these proofs (and the proofs for set agreement and the Hexagon task in the paper) are completely different from previous proofs of similar results:
round-reduction works directly on the task specification, in an algorithmic way that does not depend on the specific task at hand.
Hence, the arguments in round-reduction proofs are applied directly on the task specification, and not on executions of a protocol (which are encapsulated in the round-reduction theorem).

\subsection{Time Lower Bound for Approximate Agreement}

Let us show that the bound $\lceil\log_2D\rceil$ in Theorem~\ref{theo:connectivity} is tight.
For this purpose, consider the approximate agreement task.
For an integer $N\geq 1$, let $\epsilon=1/N$, and the $\epsilon$-agreement task defined as follows.
The input complex~$\I$ of $\epsilon$-agreement
is merely the edge
\footnote{Some works allow other input values as well, but this is not necessary for the lower bounds we present here}
\[
0 \; \rule{3em}{2pt} \; 1.
\]
The output complex~$\O$ is the path
\[
0  \; \rule{3em}{2pt} \; \epsilon  \; \rule{3em}{2pt} \; 2\epsilon  \; \rule{3em}{2pt} \; \dots  \; \rule{3em}{2pt} \; (N-1)\epsilon \; \rule{3em}{2pt} \; 1.
\]
Finally, the input-output specification $\Delta$ satisfies for each set $S
\subseteq \I$
\[
\Delta(S)=
\{T\subseteq\O \mid\min{S}\leq \min{T}\textrm{ and } \max{T}\leq\max{S}\}
\]
and specifically, for an element $x\in \I$ it specifies $\Delta(\{x\})=\{x\}$.

\begin{proposition}
For every $\epsilon\in (0,1)$, $\epsilon$-agreement cannot be solved by $n\geq 2$ processes in less than $\lceil\log_2 1/\epsilon\rceil$ rounds.
\end{proposition}

Before proving this preposition, we remark that for $n>2$ processes this bound is tight, and is the same for colored and colorless algorithm.
Interestingly, for $n=2$ processes there is a colored algorithm requiring only $\lceil\log_31/\epsilon \rceil$ rounds~\cite{AspnesH90,HoestS06}.
Hence, while colorless and colored algorithms have the same computability power, colored algorithms are provably stronger in terms of time complexity.

\begin{proof}
The diameter $D$ of $\O$ is $N=1/\epsilon$. Let us first show that if $0\leq k <\lceil\log_2 1/\epsilon\rceil$, we have $\{0,1\}\notin \Delta^{(k)}(\{0,1\})$. The proof is based on the following fact: for any two distinct vertices  $u$ and~$v$ of $\Delta^{(k)}(\{0,1\})$ at distance at least~3 in $\Skel_1(\Delta^{(k)}(\{0,1\}))$, we have $\{u,v\}\notin \Delta^{(k+1)}(\{0,1\})$. To establish this fact, let us consider any map
\[
f:\Bary^{(1)}(\{u,v\})\to \Delta^{(k)}(\{0,1\})
\]
agreeing with $\Delta^{(k)}_{\{u,v\},\{0,1\}}$. Since $f$ agrees with $\Delta^{(k)}_{\{u,v\},\{0,1\}}$, we must have $f(u)=u$ and $f(v)=v$. Therefore, if $w=f(\{u,v\})$ then either $\{u,w\}$ or $\{w,v\}$ is not an edge of $\Delta^{(k)}(\{0,1\})$ because $u$ and $v$ are at distance greater than~2. Therefore $f$ is not simplicial, which shows that $\{u,v\}\notin \Delta^{(k+1)}(\{0,1\})$, as claimed. It follows that, for $k <\lceil\log_2 1/\epsilon\rceil$, we have $\{0,1\}\notin \Delta^{(k)}(\{0,1\})$.

This latter fact implies that the $k$-th closure of $\epsilon$-agreement is not solvable in zero rounds: an algorithm~$f$ solving this $k$-th closure must satisfy $f(0)=0$ and $f(1)=1$.
As a consequence, if some processes start with~0, and some other processes start with input~1, all these processes jointly output the set $\{0,1\}$, which is not a valid output as $\{0,1\}\notin \Delta^{(k)}(\{0,1\})$.
\end{proof}

\subsection{Impossibility of Covering Tasks}
\label{subsec:covering-tasks}

Recall that, for two connected simplicial complexes $\I$ and $\O$, and for a simplicial
map $f:\O\to\I$, the pair $(\O,f)$ is a \emph{covering complex} of~$\I$ if, for every~$\sigma\in\I$, $f^{-1}(\sigma)$~is a union of pairwise disjoint simplexes.
This condition can be rephrased as $f^{-1}(\sigma)=\cup_{i=1}^k\tau_i$ with $f_{|\tau_i}:\tau_i\to \sigma$ is one-one.
The simplexes $\tau_i$, $i=1,\dots,k$, are called the \emph{sheets} of $\sigma$. We often refer to $f$ as a \emph{covering map}. The following observations follow directly from the definition of covering complex (see, e.g., \cite{Rotman73}).
\begin{itemize}
\item If $\sigma\in \I$ is a simplex of dimension $d$, each sheet $\tau_i$ of $\sigma$ is also a simplex of dimension~$d$.
\item The two complexes $\I$ and $\O$ are \emph{locally isomorphic}, in the sense that for each vertex $v\in \O$ the complex $\mbox{star}(v)$ is isomorphic to the complex $\mbox{star}(f(v))$.
Recall that the star of a vertex~$v$ in a complex~$\mathcal{K}$
is the complex $\mbox{star}(v)$ consisting of all the simplexes of~$\mathcal{K}$ that contain~$v$.
\item All the simplices in $\O$ have the same number of sheets.
\end{itemize}
We define below a colorless variant of the chromatic covering tasks introduced in~\cite{FraigniaudRT13}.

\begin{definition}
Given a covering complex $(\O,f)$ of a complex $\I$, the colorless \emph{covering task} $(\I,\O,\Delta)$ is the task where $\Delta$ is defined, for every $\sigma\in\I$, by
\[
\Delta(\sigma)=\{\tau\in\O\mid f(\tau)\subseteq \sigma\},
\]
where $f(\tau)\subseteq \sigma$ means that $f(\tau)$ is a sub-complex of the complex defined by~$\sigma$ and all its faces.
A covering complex is \emph{non-trivial} if each simplex in $\I$ has more than one sheet.
\end{definition} 
The Hexagone task discussed above is a basic example of a covering task. Figure~\ref{fig:cover-task} provides another example of a covering task;
there, $f(\sigma')=f(\sigma'')=\sigma$, and accordingly the image of $\sigma$ under $\Delta$ is the union of the two complexes with unique facets $\sigma'$ and $\sigma''$.

\begin{figure}[!t]
\centering
\includegraphics[width=12cm]{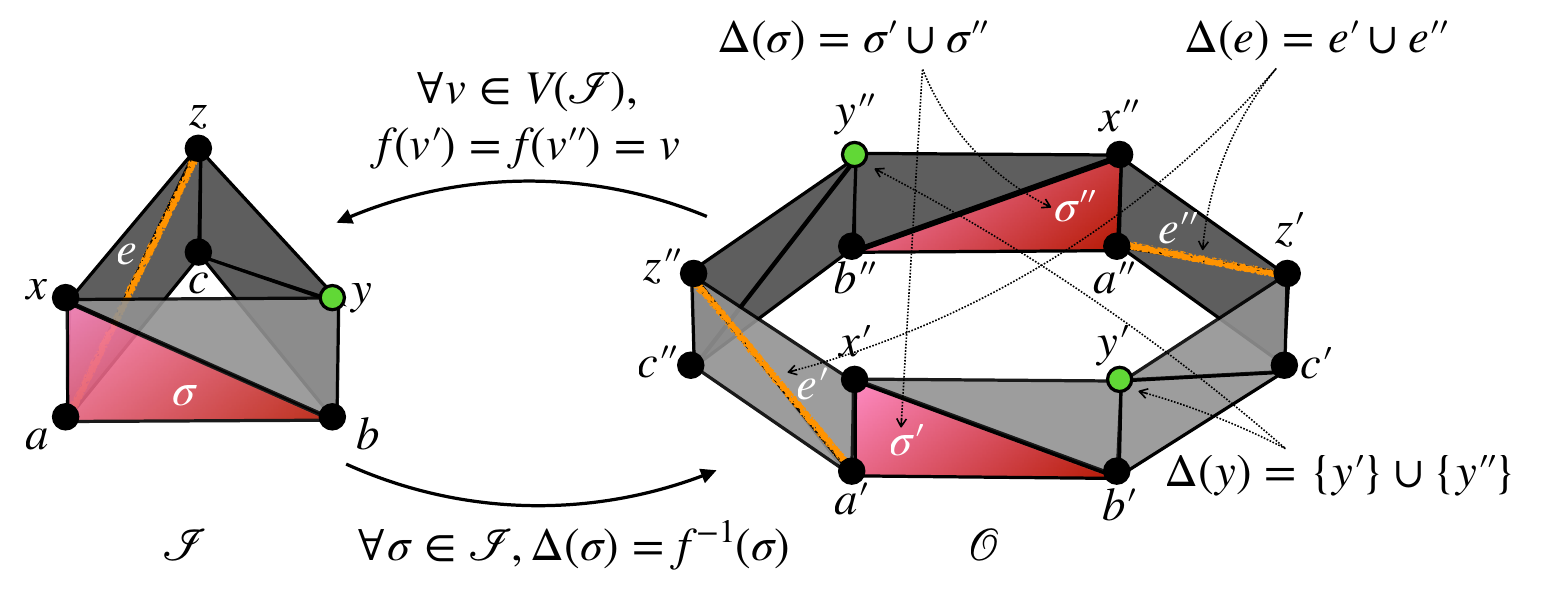}
\caption{\sl A 2-dimensional covering task extending the Hexagone task to a higher dimension}
\label{fig:cover-task}
\end{figure}

We now turn to prove a general impossibility result for covering tasks.
\begin{theorem}
	\label{thm:covering impossibility}
	No non-trivial colorless covering tasks can be solved by $n\geq2$ processes running a wait-free colorless algorithm.
\end{theorem}

A similar result was proved in the past using an ad-hoc argument~\cite{FraigniaudRT13} for colored covering tasks,
and here we give a round-reduction based proof for colorless covering tasks.
Recall that by Theorem~\ref{theo:rr-to-ex-colorless}, this also means that the claim has an FLP-style proof,
which proves Theorem~\ref{theo:D}.

\begin{proof}
	Consider a non-trivial covering complex $(\O,f)$ of a complex $\I$, and the corresponding covering task $\Pi=(\I,\O,\Delta)$.
	For the case of the Hexagon task we have seen that $\Cl(\HX)=\HX$, and here we start the same why.
	
	Consider an input simplex $\sigma\in\I$, and note that each of its sheets is a simplex, and that its sheets do not intersect.
	Hence, all the connected components of $\Delta(\sigma)$ are complete (each is of dimension $\dim(\sigma)$),
	and
	Theorem~\ref{theo:connectivity} implies that $\Delta^\ast(\sigma)=\Delta(\sigma)$, so $\Cl(\Pi)=\Pi$.
%
%
	Hence, if $\Pi$ is solvable wait-free by $n\geq2$ processes in the IIS model, then
	by Lemma~\ref{lem:0round} it is wait-free solvable in zero rounds.

	Consider a possible zero-round algorithm for $\Pi$, and its decision map~$\delta$.
	Recall that $\delta$ must be simplicial, i.e. maps simplices to simplices, and hence also paths to paths.
	
	Fix $x\in V(\I)$, its image $y=\delta(x)\in V(\O)$ under $\delta$, and note that $f(y)=x$.
	Since $(\O,f)$ is non-trivial, there is another vertex $y'\in V(\O)$, $y'\neq y$, such that $f(y')=x$.
	As $\O$ is connected, it contains a path $(y_0=y,y_1,\ldots,y_k=y')$ connecting $y$ and $y'$, and $\I$ contains its image $C=(x_0=f(y_0),\ldots,x_k=f(y_k))$ under $f$.
	Note that $x=x_0=x_k$, so $C$ is in fact a cycle in $\I$.
	Apply $\delta$ to $C$, and the fact that $\delta$ is simplicial gives a cycle $\delta(C)$ in $\O$.
	
	As $\delta(x_0)=y_0$ but $\delta(x_k)\neq y_k$, there must exist a minimal index $0\leq i<k$ such that $\delta(x_i)=y_i$ and $\delta(x_{i+1})\neq y_{i+1}$.
	By the construction of the path, $f(y_i,y_{i+1})=(x_i,x_{i+1})$, and by the assumption that $\delta$ solves the task and hence comply with $\Delta$ we have
	$f(y_i,\delta(x_{i+1}))=(x_i,x_{i+1})$.
	Hence $f^{-1}(x_i,x_{i+1})$ contains both $(y_i,y_{i+1})$ and $(y_i,\delta(x_{i+1}))$, i.e.~$(x_i,x_{i+1})$ has two intersecting sheets, in contradiction to $(\O,f)$ being a covering complex.
\end{proof}


\section{Conclusion}


The purpose of this paper is to relate
round-reduction proof techniques (formally stated in the Speedup Theorem) and FLP-style proof techniques, when applied to colorless tasks within the framework of wait-free computing in the IIS model.

The round-reduction technique offers many good features, including the fact that it is mechanical (it is sufficient to check whether the fixed-point closure is solvable in zero rounds), it enables to derive not only impossibility results but also complexity lower bounds (e.g., for approximate agreement),
and it extends to wait-free computing in models stronger than IIS (e.g., IIS augmented with \textsf{Test\&Set} objects).
On the other hand, FLP-style proofs are very generic, and essentially apply to all models, including $t$-resilient models.
Moreover, we have shown that FLP-style proofs are not weaker than round-reduction proofs,
and it is possible that they are stronger.
Nevertheless, we have also shown that for 1-dimensional colorless tasks the two techniques have exactly the same power,
and are both complete in the sense that if a task is not solvable then any of the two techniques will enable to establish this fact.

It would be interesting to know whether the equivalence between the two techniques holds for arbitrary colorless tasks, and not only for the 1-dimensional ones,
and we conjecture that this is indeed the case.
Note however that if this conjecture is true, then these two proof techniques cannot be complete, simply because it is known that set-agreement impossibility has no extension-based proof~\cite{AlistarhAEGZ19,AttiyaCR20}.

The round-reduction $\Cl$ we define and study in this work is not the only one possible 
	(nor are 
	$F_{\mbox{\tiny IIS}}$ used in~\cite{FPR2022} nor or the operator from Appendix~\ref{app:comparison-based});
	defining a more powerful operator for proving lower bounds on general algorithms is a central open question left in~\cite{FPR2022}, 
	and we leave a similar question open here with regard to colorless algorithm.
	Nevertheless, we have shown that for 1-dimensional colorless tasks, $\Cl$ is in fact the most powerful operator possible --- this operator is shown to be complete for such tasks in Corollary~\ref{cor:round-reduction-complete}.
	The more general question of whether there exists if-and-only-if operators for wait-free computing, as was shown for synchronous failure-free computing in networks, is another central open question. 

Another research direction is to try to design an analog of round-reduction for other computational models, such as  $t$-resilient models
(the speedup theorem of~\cite{FPR2022}, as well as ours, assume the ability of each process to run solo).
The ultimate goal of the line of study initiated in this paper is to better understand the relation between backward and forward induction in the context of distributed computing.


\bibliography{colorless-speedup}

\appendix

\newpage
\centerline{\Large \bf A P P E N D I X}

\section{Comparison-Based Algorithms}
\label{app:comparison-based}

In this section, we show that using a comparison-based closure operator does not suffice for obtaining a speedup theorem for comparison-based algorithms. Recall that a comparison-based algorithm~\cite{ChaudhuriHT99} is bounded not to use the actual values of the process IDs, but solely there relative values. That is, an algorithm $A$ is comparison-based if the output of each process is the same in the two scenarios in which $\ID(p_i)=x_i$ for $i=1,\dots,n$, or $\ID(p_i)=y_i$ for $i=1,\dots,n$, as long as, for every $i\neq j$, $x_i<x_j \iff y_i<y_j$. We use the same notations as in Section~\ref{sec:round-reduction-proofs}. Let us define the comparison-based closure $\Pi'=\Cl(\Pi)=(\I,\O',\Delta')$ of a general (i.e., chromatic) $n$-process task $\Pi=(\I,\O,\Delta)$ as in Definition~\ref{def:colorless-closure}, with the requirement that the local task $\Pi_{\sigma,\tau}$ must be solvable by a comparison-based algorithm (instead of a colorless algorithm). Now, assume, for the purpose of contradiction, that a speedup theorem equivalent to Theorem~\ref{theo:closure} holds for comparison-based algorithms, that is, for every $n$-process  task $\Pi=(\I,\O,\Delta)$, and every $t>0$, if $\Pi$ is wait-free solvable in $t$ rounds by a comparison-based algorithm then the comparison-based closure $\Cl(\Pi)$ is wait-free solvable in $t-1$ rounds by a comparison-based algorithm. We obtain a contradiction by considering weak symmetry-breaking.

The weak symmetry-breaking task $\Pi=(\I,\O,\Delta)$ for $n$ processes $p_1,\dots,p_n$ is defined as follows. Processes have no inputs apart from the fact that every process $p_i$ knows its ID~$i$. That is, $\I$ has a single facet $\sigma^{(n)}=\{1,\dots,n\}$. Each process must output a value in $\{0,1\}$, with the only constraint that if all processes participate, then at least one must output~0, and at least one  must output~1. That is, $\O$ is obtained from the full $(n-1)$-dimensional chromatic complex by removing the two facets
\[
\phi^{(n)}_0=\{(i,0):i=1,\dots,n\} \;\mbox{and}\; \phi^{(n)}_1=\{(i,1):i=1,\dots,n\}.
\]
Weak symmetry-breaking is trivially solvable by a general algorithm using the actual values of the process IDs. For instance $p_1$ outputs~0, while all other processes output~1 solves the task\footnote{This algorithm is not comparison-based because $p_1$ running solo outputs~0, while $p_2$ running solo outputs~1.}. However, if one requires the algorithm to be comparison-based, then the solvability of symmetry-breaking exhibits a complex panorama. Specifically,  weak symmetry-breaking is not wait-free solvable by a comparison-based algorithm when $n$ is a power of a prime, and is wait-free solvable otherwise~\cite{AttiyaP16,CastanedaR10,CastanedaR12}.
In particular it is not solvable for $n=2,3,4,5$, but is solvable for $n=6$ (actually it is solvable by a 3-round comparison-based algorithm for all $n=6k$, $k\geq 1$~\cite{Kozlov17}). We shall show that, for every $n\geq 2$, the $n$-process weak symmetry-breaking task is a fixed point for the comparison-based closure. Since, this task is not solvable in zero rounds by a comparison-based algorithm, this implies that, for every~$n\geq 2$, weak symmetry-breaking is not solvable by any comparison-based algorithm, a contradiction.

We therefore focus now on establish the following.

\begin{proposition}\label{prop:comparison-based}
Let $n\geq 2$, and let $\Pi$ be the $n$-process weak symmetry-breaking task. The comparison-based closure $\Cl(\Pi)$ of $\Pi$ satisfies $\Cl(\Pi)=\Pi$.
\end{proposition}

\begin{proof}
The proof is by induction on~$n$. We denote by $\Xi$ the chromatic subdivision corresponding to one round of a wait-free algorithm in the \IIS\/ model. For $n=2$, $\Xi(\phi^{(2)}_0)$ is the path
\[
(0,\{(0,0)\}) \; \rule{3em}{2pt} \;   (1,\{(0,0),(1,0)\}) \; \rule{3em}{2pt} \;   (0,\{(0,0),(1,0)\}) \; \rule{3em}{2pt} \;  (1,\{(0,0)\})
\]
By definition of the local task $\Pi_{\phi^{(2)}_0,\sigma^{(2)}}$, the two extremities of this path must output~0. The two central nodes may output two different values with a comparison-based algorithm, but this will not prevent one of the three edges to have its extremities mapped to the same output, i.e., to an edge not in~$\Delta(\sigma^{(2)})$. It follows that $\phi^{(2)}_0\notin\Delta'(\sigma^{(2)})$. By the same arguments, we get that $\phi^{(2)}_1\notin\Delta'(\sigma^{(2)})$. As a consequence, the proposition holds for $n=2$.

For the induction step, let us show that $\phi^{(n)}_0\notin\Delta'(\sigma^{(n)})$ for $n\geq 3$ (see Figure~\ref{fig:comparison-based}). The 1-round protocol complex $\Xi(\phi^{(n)}_0)$ the chromatic subdivision of $\phi_0$. For $i=1,\dots,n$, let $\tau^{(n-1)}_i$ be the $(n-2)$-dimensional face of $\phi_0$ where process~$i$ does not appear. The key point is that any comparison-based algorithm will act exactly the same on each of the faces $\tau^{(n-1)}_i$, $i=1,\dots,n$, that is, a comparison-based algorithm maps all complexes $\Xi(\tau^{(n-1)}_i)$, $i=1,\dots,n$, to isomorphic complexes which differ only in the IDs of the processes, but not in the outputs. By the induction hypothesis, all of these complexes contains an all-0 simplex, or all of them contains and all-1 simplex. Let us assume, w.l.o.g., that it is an all-0 simplex, and denote these simplices by $\rho^{(n-1)}_1,\dots,\rho^{(n-1)}_n$ (see Figure~\ref{fig:comparison-based}). For avoiding to create an  all-0 facet in $\Xi(\phi^{(n)}_0)$, each of the center nodes $(i,\{(j,0):j=1,\dots,n\})$, $i=1,\dots,n\}$ must thus output~1, which results in outputting an all-1 facet. Such a facet is not in $\Delta(\sigma^{(n)})$, and therefore no comparison-based can solve the local task $\Pi_{\phi^{(n)}_0,\sigma^{(n)}}$, and thus $\phi^{(n)}_0\notin\Delta'(\sigma^{(n)})$.

By the same arguments, we get $\phi^{(n)}_1\notin\Delta'(\sigma^{(n)})$, and therefore $\Cl(\Pi)=\Pi$.
\end{proof}

\begin{figure}[!t]
\centering
\includegraphics[width=7cm]{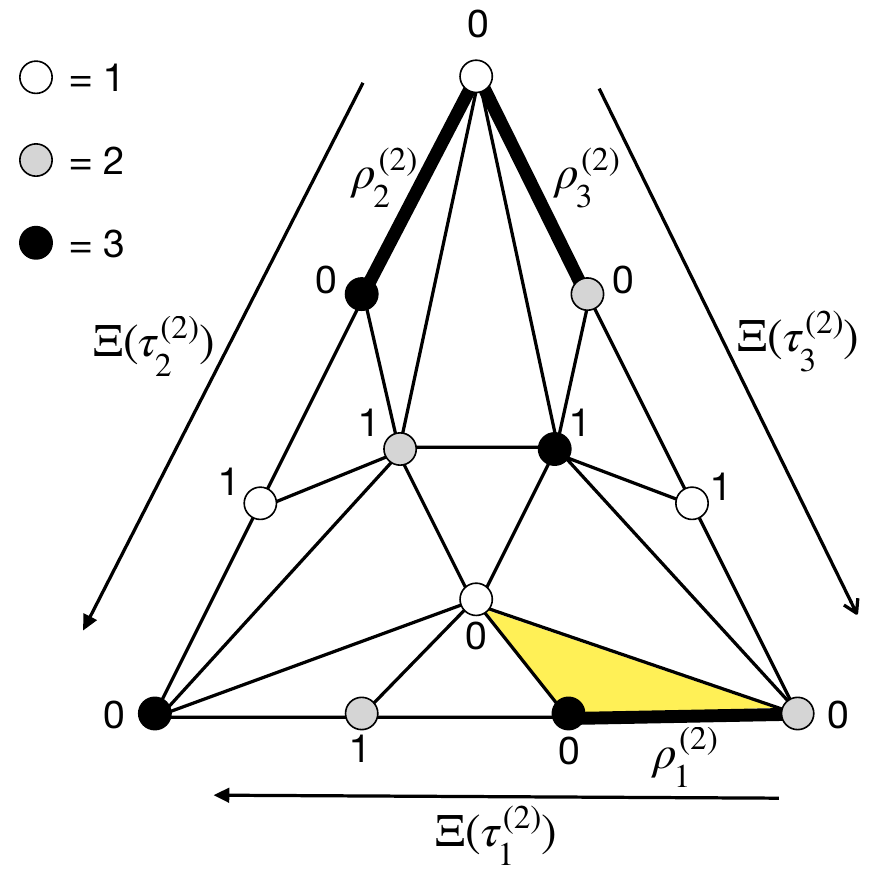}
\caption{\sl Proof of Proposition~\ref{prop:comparison-based}. Outputs of an hypothetical comparison-based algorithm solving the local task $\Pi_{\phi^{(3)}_0,\sigma^{(3)}}$. By the symmetries induced by the comparison-based algorithm on the sides of the triangle, the outputs of the three central nodes cannot prevent the existence of an all-0 triangle, or of an all-1 triangle (such as the marked bottom-right triangle). }
\label{fig:comparison-based}
\end{figure}

\end{document}